\documentclass{article}
\usepackage{amsmath,amssymb,amsthm,xypic,stmaryrd}

\title{Operational Theories and Categorical Quantum Mechanics}
\author{Samson Abramsky and Chris Heunen \\
Department of Computer Science\\
University of Oxford}

\theoremstyle{plain}
\newtheorem{theorem}[equation]{Theorem}
\newtheorem{corollary}[equation]{Corollary}
\newtheorem{lemma}[equation]{Lemma}
\newtheorem{proposition}[equation]{Proposition}

\newcommand{\Rvalue}[1]{\ensuremath{\llbracket #1 \rrbracket}}

\newcommand{\beqa}{\begin{eqnarray*}}
\newcommand{\eeqa}{\end{eqnarray*}\par\noindent}

\newcommand{\WW}{\mathcal{W}}

\newcommand{\rarr}{\rightarrow}

\newcommand{\II}{I}

\newcommand{\ket}[1]{{|} #1\rangle}
\newcommand{\bra}[1]{\langle #1{|}}
\newcommand{\ie}{\textit{i.e.}\ }
\newcommand{\eg}{\textit{e.g.}\ }
\newcommand{\CC}{\mathbf{C}}

\newcommand{\CD}{\mathcal{D}}
\newcommand{\DD}{\mathbf{D}}

\newcommand{\Rel}{\mathbf{Rel}}
\newcommand{\FRel}{\mathbf{FRel}}
\newcommand{\Set}{\mathbf{Set}}

\newcommand{\Mat}{\mathbf{FMat}}
\newcommand{\Complex}{\mathbb{C}}
\newcommand{\LMat}{\mathbf{LMat}}

\newcommand{\Hilb}{\mathbf{Hilb}}
\newcommand{\FdHilb}{\mathbf{FHilb}}

\newcommand{\id}{\mathsf{id}}

\newcommand{\dom}{\mathrm{dom}}

\newcommand{\Cat}{\mathbf{Cat}}
\newcommand{\Nat}{\mathbb{N}}

\newcommand{\IFF}{\; \Longleftrightarrow \;}

\newcommand{\IMP}{\; \Rightarrow \;}

\newcommand{\ev}{v}
\newcommand{\End}{\mathrm{End}}
\newcommand{\Ideal}{\mathcal{I}}
\newcommand{\Tr}{\mathrm{Tr}}

\newcommand{\DS}{\Delta_{S}}
\newcommand{\DT}{\Delta_{T}}

\newcommand{\Chu}{\mathbf{Chu}}
\newcommand{\embsymb}{\iota}
\newcommand{\embS}{\embsymb^{\Sa}}
\newcommand{\embM}{\embsymb^{\MM}}

\newcommand{\One}{\mathbf{1}}
\newcommand{\natarrow}{\stackrel{\cdot}{\to}}
\newcommand{\dnatarrow}{\stackrel{\cdot\cdot}{\to}}
\newcommand{\Ct}{\CC_{\mathsf{t}}}
\newcommand{\op}{\mathrm{op}}
\newcommand{\KW}{\mathbf{K}_{\WW}}

\newcommand{\MM}{\mathsf{M}}
\newcommand{\KD}{\mathrm{K}_{\CD}}
\newcommand{\dev}{\mathsf{d}}
\newcommand{\Sa}{\mathsf{P}}
\newcommand{\La}{\Lambda}
\newcommand{\Stoch}{\mathbf{Stoch}}
\newcommand{\SStoch}{\mathbf{SStoch}}
\newcommand{\ql}{q^{\lambda}}
\newcommand{\ds}{d_s}
\newcommand{\dx}{\delta_x}
\newcommand{\ddo}{\delta_o}
\newcommand{\px}{p^{x}}

\newcommand{\Dbl}[1]{\ensuremath{#1_{\leftrightarrows}}}
\newcommand{\DagCat}{\mathbf{DagCat}}
\newcommand{\qes}{\bot}

\begin{document}

\maketitle

\begin{abstract}
  A central theme in current work in quantum information and quantum
  foundations is to see quantum mechanics as occupying one point in a
  space of possible theories, and to use this perspective to
  understand the special features and properties which single it out,
  and the possibilities for alternative theories. 
  Two formalisms which have been used in this context are
  \emph{operational theories}, and \emph{categorical quantum
    mechanics}. 
  The aim of the present paper is to establish strong connections
  between these two formalisms. 
  We  show how models of categorical quantum mechanics have
  representations as operational theories. 
  We then show how \emph{non-locality} can be formulated at
  this level of generality, and study a number of examples from this
  point of view, including Hilbert spaces, sets and relations, and
  stochastic maps. The local, quantum, and no-signalling models are characterized in these terms.
\end{abstract}

\section{Introduction}

A central theme in current work in quantum information and quantum
foundations is to see quantum mechanics as occupying one point in a
space of possible theories, and to use this perspective to understand
the special features and properties which single it out, and the
possibilities for alternative theories. 

Two formalisms which have been used in this context are
\emph{operational theories}
\cite{Mack63,jauch1968foundations,piron1976foundations,ludwig1983foundations},
and \emph{categorical quantum mechanics}
\cite{abramsky2004categorical,abramsky2008categorical}. 
\begin{itemize}
  \item Operational theories allow general formulations of results in
    quantum foundations and quantum information
    \cite{barnum2007generalized,barnum2008teleportation,barnum2010entropy}. They
    also play a prominent r\^ole in current work on axiomatizations of
    quantum mechanics
    \cite{hardy2001quantum,chiribella2011informational,masanes2011derivation,dakicbrukner:entanglement}. 

  \item Categorical quantum mechanics enables a high-level approach to
    quantum information and quantum foundations, which can be
    presented in terms of string-diagram representations of structures
    in monoidal categories \cite{abramsky2008categorical}. This has
    proved very effective in providing a conceptually illuminating and
    technically powerful perspective on a range of topics, including
    quantum protocols \cite{abramsky2004categorical}, entanglement
    \cite{coecke2010compositional}, measurement-based quantum
    computing \cite{duncan2010rewriting}, no-cloning
    \cite{abrnocloning}, and non-locality \cite{coecke2011phase}. 
\end{itemize}

The aim of the present paper is to establish strong connections between these two formalisms.
We shall begin by reviewing operational theories. We then show how a
proper formulation of \emph{compound systems} within the operational
framework leads to a view of operational theories as representations
of monoidal categories of a particular form. We call these
\emph{operational representations}. 

We then review some elements of categorical quantum mechanics, and
show how \emph{monoidal dagger categories}, equipped with a trace
ideal, give rise to operational representations. 
Thus there is a general passage from categorical quantum mechanics to operational theories.

We go on to show how \emph{non-locality} can be formulated at
this level of generality, and study a number of examples from this
point of view, including Hilbert spaces, sets and relations, and
stochastic maps. The local, quantum, and no-signalling models are characterized in these terms.

We shall assume some familiarity with the linear-algebraic formalism
of quantum mechanics, and with the first notions of category
theory. To make the paper reasonably self-contained, we include an
appendix which reviews the basic definitions of monoidal categories,
functors and natural transformations. 

We also include another appendix which proves a number of technical results on trace ideals.
These are mathematically interesting, but would break up the flow of ideas in the main body of the paper.

\section{Why operational theories?}

Before proceeding to a formal description of operational theories, it
may be useful to discuss the motivation for studying them. 

As we see it, operational theories have the following attractions: 
\begin{itemize}
  \item Firstly, they focus on the empirical content of theories, and
    the means by which we can gain knowledge of the microphysical
    world. Any viable theory must account for this content. 

  \item By focussing on this empirical and observational content,
    operational theories allow meaningful results to be formulated and
    proved about the `space of theories' as a whole. At a stage in the
    development of physics where the next step is far from clear, this
    is a useful perspective, which may prove useful in finding
    `deeper' theories. 

  \item Indeed, the operational framework has proved fruitful as  a
    basis for general results, \eg on the information processing
    capabilities of theories under various assumptions
    \cite{barnum2007generalized,barnum2008teleportation,barnum2010entropy};
    and provides the setting for recent work on axiomatic
    reconstructions of quantum mechanics
    \cite{hardy2001quantum,chiribella2011informational,masanes2011derivation,dakicbrukner:entanglement}. 
\end{itemize}

On  the debit side, operational theories attract criticism on philosophical grounds.
They are seen as linked to an `instrumentalist' or `epistemic' view of
physics, as opposed to a `realistic' approach. From our perspective,
the fact that we study operational theories does not indicate any such
philosophical commitment. Rather, they are pragmatically useful for
the reasons already mentioned, and can be seen as expressing some
irreducible minimum of empirical content, which will have to be
accounted for by any presumptive `deeper' theory. 

\section{Operational theories formalized}

An operational theory is formulated in terms of directly accessible
`operations', which can be performed \eg in a laboratory. 
We assume there are several different types of system, $A$, $B$, $C$, etc.
For each system type $A$, the theory specifies the following:
\begin{itemize}
  \item A set of \emph{preparations} $P_{A}$ which produce systems of
    that type. 
  \item A set $T_{A}$ of \emph{transformations} which may be performed
    on systems of type $A$. 
    More generally, we can consider transformations $T_{A, B}$ which can
    be performed on systems of type $A$ to produce systems of type $B$. 
  \item A set of \emph{measurements} $M_{A}$ which can be performed on
    systems of that type. 
\end{itemize}
Each measurement has a set of possible \emph{outcomes}.
In this paper, we shall only consider `finite-dimensional' theories,
or parts of theories. This means that each measurement has only
finitely many possible outcomes. For convenience, we shall assume a
fixed \textit{infinite} set of outcomes $O$, which will apply to all
measurements.  Any measurement with a finite set of outcomes $O'
\subseteq O$ can be represented using $O$, where those outcomes
outside $O'$ have zero probability of occurring. 

The empirical predictions of the theory are given by its
\emph{evaluation rule}, which is a function 
\[ \ev_{A} \colon P_{A}  \times M_A \times O \to [0, 1]  \]
which assigns a  probability $\ev_{A}(p, m, o)$ 
to the event that a system of type $A$, prepared by $p$, 
yields outcome $o$ when measurement $m$ is performed on it. 

For each choice of $p$ and $m$, the function $\ev_{A}(p, m, {-})$
defines a probability distribution on outcomes. 
We shall use the function
\[ d_A \colon P_A \times M_A \to \CD 
\qquad d_A(p, m) \colon o \mapsto \ev_A(p, m, o) \]
where $\CD$ is the set of probability distributions of finite support on $O$.


\subsection{Compound systems}

An important additional ingredient is to give an account of
\emph{compound systems}, \ie putting systems, possibly space-like
separated, together. 

This leads to the following additional requirements.
\begin{itemize}
  \item For each pair of system types $A$, $B$, a compound system type
    $AB$. 
  \item Ways of combining preparations, measurements, etc.\ on $A$ and
    $B$ to yield corresponding operations on the compound system
    $AB$. 
\end{itemize}
Moreover, these operations should be subject to axioms yielding a
coherent mathematical structure on these notions. 

Rather than trying to develop such `meta-operations' and axioms from
first principles, we see the essential elements as provided by
\emph{monoidal categories}, which have been  developed extensively as
a setting for quantum mechanics and quantum information in the
categorical quantum mechanics programme
\cite{abramsky2004categorical,abramsky2008categorical}. 

We shall therefore proceed by giving a precise formulation of
operational theories with compound system structure as a certain class
of representations of monoidal categories, which we call
\emph{operational representations}.  

\subsection{Operational representations: concrete description}
\label{conoprepsec}

Before giving the `official' definition of operational representation,
which is mathematically elegant but a little abstract, we shall give a
more concrete account, which shows the naturalness of the ideas, and
also indicates why guidance from category theory is helpful in finding
the right structural axioms. 

For each system type $A$, we can gather the relevant data provided by
an operational theory into a single structure 
\[ (P_A, M_A, d_A \colon P_A  \times M_A \to \CD) . \]
This immediately suggests the notion of \emph{Chu
  space}~\cite{Barr79,Chu79}, which has received quite extensive
development~\cite{DBLP:journals/apal/Pratt99}, and was applied to the
modelling of physical systems in~\cite{btm}. 
Indeed, it can be seen as a generalization of the notion of model of a
physical system proposed by Mackey in his influential work on the
foundations of quantum mechanics \cite{Mack63}. 

There is a natural equivalence relation on preparations: $p$ is
equivalent to $p'$, where $p, p' \in P_A$, if for all $m \in M_A$: 
\[ d_A(p, m) = d_A(p', m) . \]
This is exactly the notion of extensional equivalence in Chu
spaces~\cite{btm}. We can regard \emph{states} operationally as
equivalence classes of preparations~\cite{peres1993quantum}. 

In an entirely symmetric fashion, there is an equivalence relation on
measurements. We define $m$ to be equivalent to $m'$, where $m, m' \in
M_A$, if for all $p \in P_A$: 
\[ d_A(p, m) = d_A(p, m') . \]
We can regard \emph{observables} operationally as equivalence classes
of measurements. 

Quotienting  an operational system $(P_A, M_A, d_A)$ by these
equivalences corresponds to the \emph{biextensional collapse} of a Chu
space~\cite{btm}. 

Having identified operational systems with Chu spaces, we now turn to
morphisms. A transformation in $T_{A, B}$ induces a map $f_* \colon
P_A \to P_B$. That is, preparing a system of type $A$ according to
preparation procedure $p$, and then subjecting it to a transformation
procedure $t$ resulting in a system of type $B$, is itself a procedure
for preparing a system of type $B$. 

Such a transformation can also be seen as a procedure for converting
measurements of type $B$ into measurements of type $A$: given a
measurement $m \in M_B$, to apply it to a state prepared by $p \in
P_A$, we apply the transformation $t$ to obtain a preparation of type
$B$, to which $m$ can be applied. 
Thus we can also associate  a map $f^* \colon M_B \to M_A$ to the
transformation $t$. 
The formal relationship that links the two maps $f_*$ and $f^*$ is
that, whether we measure $f_*(p)$ with $m$, or $p$ with $f^*(m)$, we
should observe the same probability distribution on outcomes: 
\begin{equation}
\label{chumcond}
d_B(f_*(p), m) = d_A(p, f^*(m)) . 
\end{equation}
This can be seen as an abstract form of the relationship between the
Schr\"odinger and Heisenberg `pictures' of quantum dynamics.  

The equation~(\ref{chumcond})  says exactly that the pair of maps
$(f_*, f^*)$ defines a morphism of Chu spaces  
\[ (f_*, f^*) \colon (P_A, M_A, d_A) \to (P_B, M_B, d_B) . \]
 Thus we see that in an entirely natural way, we can associate an
 operational theory with a sub-category of Chu spaces, more precisely
 of $\Chu(\Set, \CD)$ \cite{DBLP:journals/apal/Pratt99}. 
This sub-category will not in general be full, since not every Chu
morphism will arise from a transformation in the theory. 

However, this does not yet provide an account of compound systems. 
While Chu spaces have a standard monoidal structure, and indeed form
$*$-autonomous categories \cite{Chu79}, we should not in general
expect that operational theories will give rise to \emph{monoidal}
sub-categories of Chu spaces. Rather, we should see the notion of
compound system as an important degree of freedom, which is to be
specified by the theory. 

Thus given operational systems $A = (P_A, M_A, d_A)$ and $B = (P_B,
M_B, d_B)$, we should be able to form a system $A\otimes B = (P_{A
  \otimes B}, M_{A \otimes B}, d_{A \otimes B})$. 

What general properties should such a notion satisfy? 
One important requirement, which appears in one form or another in the
various formulations of operational theories, is to have an
\emph{inclusion of pure tensors}.  
This is given by maps
\[ 
\embS_{A, B} \colon P_A \times P_B \to P_{A \otimes B} , \qquad
\embM_{A, B} \colon M_A \times M_B \to   M_{A \otimes B}.
\]
For readability, we shall write $p \otimes p'$ rather than $\embS_{A,
  B}(p, p')$, and similarly for measurements. 

The fundamental property which this inclusion must satisfy relates to
the evaluation. For all $p \in P_A$, $p' \in P_B$, $m \in M_A$, $m'
\in M_B$, we must have: 
\begin{equation}
\label{moninctens}
d_{A \otimes B}(p \otimes p', m \otimes m') = d_A(p, m) \cdot d_B(p', m') . 
\end{equation}
This expresses the probabilistic independence of pure tensors. Conceptually, pure tensors arise by preparing states or performing measurements independently on subsystems.

In addition, there are a number of coherence conditions which are
needed to get a mathematically robust notion. Rather than writing
these down in an ad hoc fashion, we shall now turn to a more
systematic way of defining the categorical structure of operational
theories, in which these conditions arise naturally from standard
notions. 

\subsection{Operational representations: functorial formulation}

We shall now take a different view, in which the structure of an
operational theory arises from a symmetric monoidal category, which we
think of as a \emph{process category}. The operational theory will
amount to a certain form of \emph{representation} of this process
category. The receiving category for the  representation will be
$(\Set, {\times}, \One)$, viewed as a symmetric monoidal category.  

Given a symmetric monoidal category $\CC$, an operational
representation of $\CC$ is specified by the following data: 
\begin{itemize}
  \item A symmetric monoidal sub-category $\Ct$ of $\CC$. This will
    usually have the same objects as $\CC$, and only those morphisms
    which correspond to admissible transformations. 
  \item A symmetric monoidal functor $\Sa \colon \Ct \to \Set$
    which represents, for each object $A$ of $\Ct$, viewed as a type
    of system, the corresponding set of preparations or states. 
  \item A contravariant symmetric monoidal functor $\MM \colon
    \Ct^{\op} \to \Set$ which for each $A$ represents the measurements
    on $A$. Note that $\Ct^{\op}$ is a symmetric monoidal category. 
  \item A dinatural symmetric monoidal transformation 
    \[ \dev \colon \Sa \times \MM \dnatarrow \KD   \]
    which gives the evaluation rule of the theory.
    Here $\KD$ is the constant functor valued at $\CD$.
    Note that a constant symmetric monoidal functor valued at a set
    $M$  is just a commutative monoid $(M, {\cdot}, 1)$ in $\Set$. We
    take $\CD$ to be a commutative monoid under pointwise
    multiplication. 
\end{itemize}
We shall assume that the functors $\Sa$, $\MM$ are embeddings, \ie
injective on objects and faithful. 

Let us now unpack this definition.

\begin{itemize}
  \item The general point of view is that the structure of the
    operational theory is controlled by the `abstract' category
    $\CC$. The types of the theory are the objects of $\Ct$. 
  \item Rather than a single set  of preparations, we have a \emph{variable set} $\Sa$, which for each
    type $A$ gives us a set $\Sa_A$. Moreover, this acts
    functorially on the admissible transformations $f \colon A \to B$
    in $\Ct$ to produce functions $f_* \colon \Sa_A \to \Sa_B$,
    where $f_* := \Sa(f)$. 
    Thus these functions take preparations on $A$ to preparations on
    $B$, as already discussed. 
  \item Similarly, the functor $\MM$ specifies a variable set $\MM_A$ of measurements for each system type
    $A$. The contravariant action of this functor is again as expected
    from our previous discussion. 
\end{itemize}

The first new ingredient which picks up the issue of monoidal
structure is that $\Sa$ and $\MM$ are required to be \emph{monoidal}
functors. 
The fact that $\Sa$ and $\MM$ are  monoidal means that there are
natural transformations 
\[  
\embS_{A, B} \colon \Sa_A \times \Sa_B \to \Sa_{A \otimes B} , \qquad
\embM_{A, B} \colon \MM_A \times \MM_B \to   \MM_{A \otimes B}.
\]
\ie inclusions of pure tensors. Naturality means that the diagrams
\[
\xymatrix{\Sa_A \times \Sa_B \ar_-{f_* \times g_*}[d]
  \ar^-{\embS_{A,B}}[r] & \Sa_{A \otimes B} \ar^-{(f \otimes g)_*}[d] \\
  \Sa_{A'} \times \Sa_{B'} \ar_-{\embS_{A',B'}}[r] & \Sa_{A' \otimes
    B'}}
\qquad
\xymatrix{\MM_A \times \MM_B \ar^-{\embM_{A,B}}[r] 
  & \MM_{A \otimes B} \\
  \MM_{A'} \times \MM_{B'} \ar^-{f^* \times g^*}[u] \ar_{\embM_{A',B'}}[r]
  & \MM_{A' \otimes B'}  \ar_-{(f \otimes g)^*}[u]
}
\]
commute. The coherence conditions for monoidal natural transformations
complete the required properties of pure tensors. 

The dinatural transformation $\dev_{A} \colon \Sa_A \times \MM_A \to
\CD$ represents the evaluation function. Dinaturality says that for
each admissible transformation $f \colon A \to B$:
\[\xymatrix{
  & \Sa_B \times \MM_B \ar^-{\dev_B}[dr] \\
  \Sa_A \times \MM_B \ar^-{f_* \times 1_B}[ur] \ar_-{1_A \times f^*}[dr] 
  && \CD \\
  & \Sa_A \times \MM_A \ar_-{\dev_A}[ur]
}\]
Thus we see that dinaturality is exactly the Chu morphism condition~(\ref{chumcond}).
Monoidality of $\dev$ is the equation~(\ref{moninctens}).

\subsection{Operational categories}

If we are given an operational representation $(\CC, \Ct, \Sa, \MM,
\dev)$ we can construct from this a single category, recovering the picture
given in Section~\ref{conoprepsec}. 

For each object $A$ of $\CC$, we have the Chu space $(\Sa_A, \MM_A,
\dev_A)$. By dinaturality of $\dev$, each morphism $f \colon A \to B$
gives rise to a Chu morphism  
\[ (f_*, f^*) \colon (\Sa_A, \MM_A, \dev_A) \to (\Sa_B, \MM_B, \dev_B) . \]
By functoriality of $\Sa$ and $\MM$, we obtain a sub-category of Chu
spaces. 

Moreover, since $\Sa$ and $\MM$ are embeddings, we can push the
symmetric monoidal structure on $\CC$ forward to this sub-category: 
\[ \Sa_A \otimes \Sa_B := \Sa_{A \otimes B}, \quad \MM_A \otimes \MM_B := \MM_{A \otimes B}, \quad f_* \otimes f'_* := (f \otimes f')_*, \quad f^* \otimes f'^* := (f \otimes f')^{*} . \]

Thus we obtain a symmetric monoidal category, whose underlying
category is a sub-category of Chu spaces. We call this the
\emph{operational category} arising from the operational
representation. 

\subsection{Generalized representations}

The structural properties of operational representations and
categories are independent of the particular choice of the monoid
$\CD$ used in specifying the dinatural transformation $\dev$. 

We shall define a \emph{generalized operational representation with
  weights $\WW$}, where $(\WW, {\cdot}, 1)$ is a commutative monoid
with a zero element, to be a tuple $(\CC, \Ct, \Sa, \MM, \dev)$, where
$\dev$ now has the form 
\[ \dev \colon \Sa \times \MM \dnatarrow \KW   \]
and $\KW$ is the constant symmetric monoidal functor valued at
$\WW$. This yields the definition of operational representation given previously when $\WW = \CD$.

We now have a general scheme for representing symmetric monoidal
categories as operational categories. So far, however, we have no
examples. We shall now show how monoidal dagger categories give rise
to operational representations in a canonical fashion, following the
ideas of categorical quantum mechanics~\cite{abramsky2008categorical}. 

\section{Monoidal dagger categories}

Monoidal dagger categories are the basic structures used in
categorical quantum mechanics~\cite{abramsky2008categorical}. 
We shall briefly review the definitions, and give a number of examples.

A \emph{dagger category} is a category $\CC$ equipped with an
identity-on-objects, contravariant, strictly involutive
functor. Concretely, for each arrow $f \colon A \to B$, there is an
arrow $f^{\dagger} \colon B \to A$, and this assignment satisfies: 
\[ 1^{\dagger} = 1, \qquad (g \circ f)^{\dagger} = f^{\dagger} \circ g^{\dagger}, \qquad f^{\dagger\dagger} = f \, . \]
We define an arrow $f \colon A \to B$ in a dagger category to be a
\emph{dagger-isomorphism} if: 
\[ f^{\dagger} \circ  f = 1_A, \qquad f \circ f^{\dagger} = 1_B . \]
A \emph{symmetric monoidal dagger category} is a dagger category with
a symmetric monoidal structure $(\CC , \otimes , \II, \lambda, \rho,
\alpha, \sigma)$ such that 
\[ (f \otimes g)^{\dagger} = f^{\dagger} \otimes g^{\dagger} 
\]
and moreover the natural isomorphisms $\lambda$, $\rho$, $\alpha$,
$\sigma$ are componentwise dagger-isos. 

\paragraph{Examples}

\begin{itemize}
  \item The category $\Hilb$ of 
    Hilbert spaces and bounded linear
    maps, and its (full) sub-category $\FdHilb$ of finite-dimensional
    Hilbert spaces. Here the dagger is the adjoint, and the tensor
    product has its standard interpretation for Hilbert spaces. 
    More generally, any symmetric monoidal C*-category is an
    example~\cite{ghezlimaroberts:wstarcategories,
      doplicherroberts:duality}. This includes categories of (right)
    Hilbert C*-modules, which are Hilbert spaces whose 
    inner product takes values in an arbitrary C*-algebra instead of
    $\Complex$.  

  \item The category $\Rel$ of sets and relations. Here the dagger is
    relational converse, while the monoidal structure is given by the
    cartesian product. This generalizes to relations valued in a
    commutative quantale~\cite{rosenthal:quantales}, and to the
    category of relations for any regular category~\cite{Butz}. Small
    categories as objects and profunctors as morphisms behave very
    similarly to $\Rel$, even though they only form a bicategory~\cite{benabou:prof}.

  \item A common generalization of $\FdHilb$ and $\FRel$, the category
    of finite sets and relations, is obtained by forming the category
    $\Mat(S)$, where $S$ is a commutative semiring with
    involution. $\Mat(S)$ has finite sets as objects, and maps $X
    \times Y \to S$ as morphisms, which we think of as `$X$ times $Y$
    matrices'. Composition is by matrix multiplication, while the
    dagger is conjugate transpose, where conjugation of a matrix means
    elementwise application of the involution on $S$. The tensor
    product of $X$ and $Y$ is given by $X \times Y$, with the action
    on matrices given by componentwise multiplication. (This
    corresponds to the `Kronecker product' of matrices). If we take $S
    = \Complex$, this yields a category equivalent to $\FdHilb$, while
    if we take $S$ to be the Boolean semiring $\{0, 1\}$ (with trivial
    involution), we get $\FRel$. 

  \item An infinitary generalization of $\Mat(\Complex)$ is given by
    $\LMat$. This category has arbitrary sets as objects, and as
    morphisms matrices $M\colon X \times Y \to \Complex$ such that for
    each $x \in X$, the family $\{ M(x, y) \}_{y \in Y}$ is
    $\ell_2$-summable; and for each $y \in Y$, the family $\{ M(x, y)
    \}_{x \in X}$ is $\ell_2$-summable. $\Hilb$ is equivalent to a
    (non-full) sub-category of $\LMat$. 


  \item If $\CC$ and $\DD$ are symmetric monoidal dagger
    categories, then so is the category $[\CC, \DD]$ of functors $F
    \colon \CC \to \DD$ that preserve the dagger. Morphisms are
    natural transformations. 
    This accounts for several interesting models. For example, setting
    $\DD=\FdHilb$ and letting $\CC$ be a group, we obtain the category
    of unitary representations. Any topological or conformal quantum
    field theory is a sub-category of the case where $\DD=\FdHilb$
    and $\CC$ is the category of cobordisms~\cite{kock:frobenius,
      atiyah:tqft, segal:cft}. 
    Letting $\CC$ be the discrete category $\mathbb{N}$, and letting
    $\DD$ be either $\FdHilb$ or $\FRel$, we recover $\Mat(\DD(I,I))$.
\end{itemize}

\paragraph{The doubling construction}

All of the above examples are variations on the theme of matrix categories.
Indeed, it seems hard to find natural examples which are not of this form.
However, there is a construction which produces a symmetric monoidal dagger category from \emph{any} symmetric monoidal category. Although the construction is formal, it is interesting in our context since it can be seen as a form of \emph{quantization}; it converts classical process categories into a form in which quantum constructions are meaningful.

Given a category $\CC$, we define a dagger category $\Dbl{\CC}$ as
follows. The objects are the same as those of $\CC$, and a morphism
$(f, g) \colon A \rarr B$ is a pair of $\CC$-morphisms $f \colon A
\rarr B$, $g \colon B \rarr A$. Composition is defined componentwise;
while $(f, g)^{\dagger} = (g, f)$. 
This is in fact the object part of the right adjoint to the evident
forgetful functor $\DagCat \rarr \Cat$; see~\cite[3.1.17]{heunen:thesis}. Thus for each
dagger category $\CC$, there is a dagger functor $\eta_{\CC} \colon
\CC \rarr \Dbl{\CC}$ which is the identity on objects, and sends
$f$ to $(f, f^{\dagger})$. This has the universal property with
respect to dagger functors $\CC \rarr \Dbl{\DD}$ for categories
$\DD$. 

This cofree construction of a dagger category lifts to the level of
symmetric monoidal categories. 
If $\CC$ is a symmetric monoidal category, then $\Dbl{\CC}$ is a
symmetric monoidal dagger category, with the monoidal structure
defined componentwise: thus $(f, g) \otimes (h, k) := (f \otimes h, g
\otimes k)$. Note in particular that the structural isos in $\CC$ turn
into \emph{dagger} isos in $\Dbl{\CC}$. 

\subsection{Additional structure}
\label{subsec:additionalstructure}

We shall require two further structural ingredients.
The first is \emph{zero morphisms}: for each pair of objects $A$, $B$, a morphism $0_{A, B} \colon A \to B$ such that, for all $f \colon C \to A$ and $g \colon B \to D$,
\[ 0_{A, B} \circ f = 0_{C, B}, \qquad g \circ 0_{A, B} = 0_{A, D} . \]
Note that if zero morphisms exist, they are unique. 

In the context of symmetric monoidal dagger categories, we further require that
\[ f \otimes 0 = 0 = 0 \otimes g, \qquad 0^{\dagger} = 0 . \]

\paragraph{Examples} All the examples of symmetric monoidal dagger
categories given above have zero morphisms in an evident
fashion. Functor categories have componentwise zero morphisms. Zero
morphisms in $\Dbl{\CC}$ are pairs of zero morphisms in $\CC$.
For more examples, see~\cite{heunenjacobs:dagkercats}.

The final ingredient we shall require is a \emph{trace ideal} in the 
sense of \cite{ abramskyetal:nuclear}.\footnote{Strictly speaking, we
  are defining the more restricted notion of \emph{global trace} of an
  endomorphism, rather than a \emph{parameterized trace} as in~\cite{
    abramskyetal:nuclear}. This restricted notion is all we shall need.} 
Firstly, we recall that in any monoidal category, the \emph{scalars},
\ie the endomorphisms of the tensor unit $I$, form a commutative
monoid \cite{kellylaplaza:compactcategories}. 

An \emph{endomorphism ideal} in a symmetric monoidal category
$\CC$ is specified by a set $\Ideal(A) \subseteq \End(A)$ for each
object $A$, where $\End(A) = \CC(A, A)$ is the set of endomorphisms on
$A$. This is subject to the following closure conditions: 
\[ g \colon A \to B, f  \in \Ideal(A), h \colon B \to A \IMP g \circ f \circ h \in \Ideal(B)  \]
\[ f \in \Ideal(A), g \in \Ideal(B) \IMP f \otimes g \in \Ideal(A \otimes B), \qquad \Ideal(I) = \End(I)  \]
\[ 0 \in \Ideal(A) . \]
If $\CC$ is a dagger category, $\Ideal$ is a \emph{dagger endomorphism
  ideal} when additionally
\[ f \in \Ideal(A) \IMP f^{\dagger} \in \Ideal(A), \]
but we will also call these endomorphism ideals for short.
 A \emph{trace ideal} is an endomorphism ideal $\Ideal$, together with a function
\[ \Tr_A \colon \Ideal(A) \to \End(I) \]
for each object $A$, subject to the following axioms:
\[ \Tr_A(g \circ f) = \Tr_B(f \circ g) \qquad(f \colon A \to B, g \colon B \to A, g \circ f \in \Ideal(A), f \circ g \in \Ideal(B)) \]
\[ \Tr_{A \otimes B}(f \otimes g) = \Tr_A(f) \Tr_B(g), \qquad \Tr_{I}(s) = s. \]
A \emph{dagger trace ideal} additionally satisfies
\[ \Tr_A(f^{\dagger}) = \Tr_A(f)^{\dagger}, \]
but we will also call these trace ideal for short.
 We call a morphism $f \in \Ideal(A)$ \emph{trace class}.

\paragraph{Examples}
All of the examples given above have trace ideals. In the case of
finite matrices, the usual matrix trace is a total operation. In the
case of $\Hilb$, we interpret trace class in the standard sense for
Hilbert spaces, and similarly for $\LMat$.
Through the
GNS-embedding~\cite[Proposition~1.14]{ghezlimaroberts:wstarcategories},
this also provides a trace ideal for any C*-category.  

In the case of relations, the summation over the diagonal becomes a
supremum in a complete semilattice, which is always defined. 

Any symmetric monoidal dagger sub-category of $[\CC,\DD]$ inherits
endomorphism ideals and zero 
morphisms from $\DD$ componentwise, and has a trace function $\Tr(\alpha) = \bigvee_A
\Tr(\alpha_A)$ as soon as $\DD(I,I)$ has an operation $\bigvee$
satisfying $\bigvee_A s_A^\dag = (\bigvee_A s_A)^\dag$, $\bigvee_A s =
s$, and $(\bigvee_A s_At_A) = (\bigvee_A s_A)(\bigvee_A t_A)$, where
$A$ ranges over the objects of $\CC$. This is the case when $\CC$ is a
finite group, as well as for topological quantum field theories.

The doubling construction turns trace ideals into dagger trace
ideals. For $(f, g) \colon A \rarr A$, define $(f, g) \in \Ideal(A)$ if and only if $f \in
\Ideal(A)$ and $g \in \Ideal(A)$, and $\Tr_A(f, g) = (\Tr_A(f),
\Tr_A(g))$. 
Thus if $\CC$ is a symmetric monoidal category with zero morphisms and
a trace ideal, $\Dbl{\CC}$ is a dagger category
with the same structure. 


In Appendix B, we prove a number of results about trace ideals:
\begin{itemize}
\item We characterize when trace ideals
exist, and to what extent they are unique.
\item We show that we really 
need to restrict to ideals to consider traces: the category of Hilbert
spaces does not support a trace on all morphisms. 
\item As a 
corollary, we derive that dual objects in the category of Hilbert
spaces are necessarily finite-dimensional. 
\item Finally, we prove
in some detail that the category of 
Hilbert spaces indeed has a
trace ideal; the details turn out to be quite subtle. 
\end{itemize}

This material would have unduly interrupted the main flow of the paper, but is of mathematical interest in its own right.

\section{From categorical quantum mechanics to operational categories}
\label{cqmoprepsec}

Let $\CC$ be a symmetric monoidal dagger category with zero morphisms
and a trace ideal. We shall show that $\CC$ gives rise to an
operational representation and operational category in a canonical
fashion, directly inspired by quantum mechanics. 

\subsection{Transformations}

We take $\Ct$ to be the sub-category with the same objects as $\CC$,
and with dagger-isomorphisms as arrows. This is a groupoid, \ie all
morphisms are invertible. 

It is easily seen to be a monoidal dagger sub-category of $\CC$.

\subsection{States}

A morphism $f \in \End(A)$ in a dagger category is \emph{positive} if
for some $g \colon A \to B$, $f = g^{\dagger} \circ g$. We define a
\emph{state} on $A$ to be a positive morphism $f \in \End(A)$ which is
trace class, and such that $\Tr_A(f) = 1$. We write $\Sa_A$ for the
set of states on $A$. 

In $\Hilb$, this definition yields exactly the standard notion  of
\emph{density operator} as used in quantum mechanics.  

Pure states can also be defined in this setting. An arrow $\psi \colon
I \to A$ has \emph{unit norm} if $\psi^{\dagger} \circ \psi =
1$. Given such an arrow, $\psi \circ \psi^{\dagger} \in
\Sa_A$. Indeed, this arrow is clearly positive, and 
\[ \Tr_A(\psi \circ \psi^{\dagger}) = \Tr_I(\psi^{\dagger} \circ \psi) = \Tr_I(1) = 1 \]
using our assumption on $\psi$ and the axioms for the trace. 

Given a dagger isomorphism $f \colon A \to B$ in $\CC$,
the function $f_* \colon \Sa_A \to \Sa_B$ is defined by
\[ f_* \colon s \mapsto f \circ s \circ f^{\dagger} . \]
Functoriality holds, since
\[ g_* \circ f_*(s) = g \circ (f \circ s \circ f^{\dagger}) \circ g^{\dagger} = (g \circ f) \circ s \circ (g \circ f)^{\dagger} = (g \circ f)_* (s). \]
Inclusion of pure tensors is given by
\[ \embS_{A,B} \colon (s, t) \mapsto s \otimes t . \]
It is straightforward to check the coherence conditions. 

\subsection{Measurements}
\label{meassec}

A \emph{dagger idempotent}, or \emph{projector}, on $A$ is an arrow $P
\in \End(A)$ such that 
\[ P^2 = P, \qquad P = P^{\dagger} . \]
A family $\{ f_i \}_{i \in I}$ of endomorphisms on $A$ is: 
\begin{itemize}
  \item \emph{Pairwise disjoint} if $f_i \circ f_j = 0$, $i \neq j$;
  \item \emph{Jointly monic} if for all $g, h \colon B \to A$:
    \[ [\, \forall i \in I. \, f_i \circ g = f_i \circ h \,] \IMP g = h . \]
\end{itemize}
A \emph{projective measurement} on $A$ with finite set of outcomes $O'
\subseteq O$ is a family of dagger idempotents $\{ P_o \}_{o \in O'}$
on $A$ which is pairwise disjoint and jointly monic. We take $\MM_A$
to be the set of projective measurements on $A$. 

The functorial action of the measurement functor on dagger
isomorphisms $f \colon A \to B$ in $\CC$ is defined by 
\[ f^*(P_o) = f^{\dagger} \circ P_o \circ f . \]
It is easily verified that $f^*$ preserves disjointness and joint
monicity of families of projectors, and hence carries projective
measurements to projective measurements. 
Functoriality  is also easily verified.

Inclusion of tensors is defined pointwise on projectors:
\[ \embS_{A,B} \colon (P_o, P_{o'}) \mapsto P_o \otimes P_{o'}. \]
Note that the combined measurement will have a finite set of outcomes
which, perhaps with some relabelling, can be regarded as a subset of
$O$. 

\subsection{Evaluation}

The transformation $\dev$ is defined as follows, where $s \in \Sa_A$,
and $m = \{ P_o \}_{o \in O'} \in\MM_A$: 
\[ \dev_{A}(s, m)(o) \; := \; \left\{ \begin{array}{ll}
\Tr_{A}(s \circ P_{o}), & o \in O' \\
0, & \mbox{otherwise.}
\end{array}
\right.
\]
Note that $\dev$ is valued in the commutative monoid of scalars $\WW
:= \End(I)^O$. By the assumption of zero morphisms, this monoid has a
zero element. 

The dinaturality of this transformation, \ie the Chu morphism
condition, is just: 
\[ \Tr_B(f \circ s \circ f^{\dagger} \circ P_o) = \Tr_A(s  \circ f^{\dagger} \circ P_o \circ f) . \]
The monoidality of $\dev$ is verified as follows: 
\begin{align*}
  \dev_{A \otimes B}(s \otimes s', m \otimes m')(o, o')
  & = \Tr_{A \otimes B}(s \otimes s' \, \circ \, P_o \otimes P_{o'}) \\
  & = \Tr_{A \otimes B}(s \circ P_o \, \otimes \, s' \circ P_{o'}) \\
  & = \Tr_A(s \circ P_o) \Tr_B(s' \circ P_{o'})  \\
  & = \dev_A(s, m)(o) \cdot \dev_B(s', m')(o') . 
\end{align*}

\subsection{The canonical operational representation}

We collect the constructions described in this section together. Given
a symmetric monoidal dagger category $\CC$ with zero morphisms and a trace
ideal, we have defined 
a sub-category $\Ct$, monoidal functors $\Sa$ and $\MM$, and a dinatural transformation $\dev$.

\begin{proposition}
  The tuple $(\CC, \Ct, \Sa, \MM, \dev)$ is an operational
  representation with weights $\WW$. We call this the \emph{canonical
  operational representation} of $\CC$. The corresponding
  operational category is the \emph{canonical operational category}
  for $\CC$. 
  \qed
\end{proposition}

We say that the canonical representation is \emph{distributional} if
the monoid of scalars $\End(I)$ has an addition making it a commutative
semiring, and for each state $s \in \Sa_A$ and measurement  
$m \in\MM_A$:
\begin{equation}
\label{disteq}
\sum_{o \in O} \dev_A(s,m)(o) \; = \; 1 . 
\end{equation}
We say that it is \emph{probabilistic} if moreover the image of $\dev$
embeds into the semiring of non-negative reals. 

\section{Examples of operational categories}

We shall now examine the operational categories arising from various
examples of symmetric monoidal dagger categories. 

\subsection{Hilbert spaces}

The definitions of states, measurements and evaluation are directly
inspired by those used in the standard Hilbert-space formulation of
quantum mechanics. Thus it is immediate that the states in the
canonical representation for $\Hilb$ are the density matrices, while
the dagger-isomorphisms are the unitary transformations. 

For measurements, we have the following result.

\begin{proposition}
  Measurements in $\Hilb$ have exactly their standard meaning. 
  More precisely, observables with finite discrete spectra  correspond
  exactly to the interpretation in $\Hilb$ of the abstract notion of
  measurements as defined in Section~\ref{meassec} for dagger
  categories. 
\end{proposition}
\begin{proof}
  We think of the outcomes as labelling the eigenvalues of the
  observable; then the family $\{ P_o \}_{o \in O'}$ should correspond
  to the \emph{spectral decomposition} of the observable. Clearly,
  dagger idempotents correspond exactly to projectors  in $\Hilb$, and
  so does the notion of a pairwise disjoint family of projectors. 
  It remains to show that the joint monicity condition captures the
  fact that a pairwise disjoint family of projectors $\{ P_i \}_{i \in
    I}$ yields a \emph{resolution of the identity}, \ie 
  \[ \sum_{i \in I} P_i = 1_A . \]
  Indeed, if $\sum_{i \in I} P_i = 1_A$ and $P_i \circ g = P_i \circ h$ for all $i$, then
  \[  g = 1_A \circ g = (\sum_{i \in I} P_i) \circ g = \sum_{i \in I} P_i \circ g = \sum_{i \in I} P_i \circ h = (\sum_{i \in I} P_i) \circ h = 1_A \circ h = h . \]
  For the converse, suppose that $\sum_{i \in I} P_i \neq 1_A$. This
  implies that for some non-zero vector $\psi$, $P_i(\psi) = 0$ for all
  $i$. Then for $f \colon \Complex \to A$ given by $1 \mapsto \psi$, we
  have $P_i \circ f = P_i \circ 0$ for all $i$, so the family is not jointly
  monic. 
\end{proof}

Finally, the definition of $\dev$ matches the standard statistical
algorithm of quantum mechanics. Thus we obtain the standard
interpretations of states, transformations, (projective) measurements,
and probabilities of measurement outcomes. 

The operational category arising from $\Hilb$ is of course
probabilistic.

The same analysis holds for C*-categories through their
GNS-construction, and for subcategories of $[\CC,\Hilb]$ such as
topological quantum field theories. States and measurements in such
categories are just natural
transformations whose components are states or measurements respectively. Because the
tensor unit in such categories is the constant functor $K_I$, they
have the same scalars as $\Hilb$. Therefore the induced operational
categories are probabilistic.

\subsection{Relations}

We shall now give a general analysis of the operational representation
for locale-valued relations. This level of generality will be useful
when we go on look at non-locality in operational categories. 

We recall that a \emph{locale} \cite{johnstone1986stone} (also known
as a \emph{frame} or \emph{complete Heyting algebra}) 
is a complete lattice $\Omega$  such that the following distributive law holds:
\[ a \wedge \bigvee_{i \in I} b_i \; = \; \bigvee_{i \in I} a \wedge b_i . \]
The category $\Rel(\Omega)$ has sets as objects, while the morphisms $R : X
\rarr Y$ are $\Omega$-valued relations (or matrices) $R : X \times Y \rarr
\Omega$. We write $\Rvalue{ xRy } = \omega$ for $R(x,y)=\omega$. 
Composition is relational composition (or matrix multiplication)
evaluated in $\Omega$. If $R \colon X \rarr Y$ and $S \colon Y \rarr Z$,
then: 
\[ \Rvalue{x (S \circ R) z} \; := \; \bigvee_{y \in Y} \Rvalue{x R y}
\wedge \Rvalue{y S z} . \]
Clearly, $\Rel$ is the special case that $\Omega$ is the Boolean
semiring $\{\bot, \top\}$, where we identify $\bot$, the bottom element of the lattice, with $0$, and $\top$, the top element, with $1$.
Note that the full subcategory $\FRel(\Omega)$ of finite sets
is identical to $\Mat(\Omega)$, where we
regard $\Omega$ as a semiring with  idempotent addition and
multiplication. Indeed, in the finite case, completeness of $\Omega$ need
not be assumed, and we are simply in the case of matrices over
idempotent semirings. 

We shall take the tensor unit in $\Rel(\Omega)$ to be $I = \{ \bullet \}$.

By an $\Omega$-subset of a set $X$, we mean a function $X \to \Omega$. 
Any family $\{S_i\}$ of $\Omega$-subsets of $X$ has a `union' $\bigvee_i
S_i$ given by $x \mapsto \bigvee_i S_i(x)$, and an `intersection'
$\bigwedge_i S_i$ given by $x \mapsto \bigwedge_i S_i(x)$. In
particular, we write $\top_X$ for the $\Omega$-subset of $X$ given by $x \mapsto
\top$, and $\qes_X$ for the $\Omega$-subset of $X$ given by $x
\mapsto \bot$.
Given a set $X$, we say that a family $\{ S_{i} \}_{i \in I}$ of
$\Omega$-subsets of $X$ is a \emph{disjoint cover} of $X$ if: 
\[ S_{i} \wedge S_{j} = \qes_X \;\; (i \neq j),  \qquad \bigvee_{i \in I} S_{i} = \top_X . \]
Given a $\Omega$-subset $S$ of $X$, we define a $\Omega$-relation $\Delta_S
\colon X \to X$ by 
\[ \Rvalue{x \Delta_{S} y } = \left\{ \begin{array}{ll} S(x) & \mbox{ if }
    x=y, \\ \bot & \mbox{ if }x \neq y. \end{array}\right. \]
Note that
\begin{equation}
\label{dcovereq}
\DS \circ \DT = \qes_{X \times X} \IFF S \wedge T = \qes_X, \qquad \bigvee_{i \in I} \Delta_{S_{i}} = 1_{A} \IFF \bigvee_{i \in I} S_{i} = \top_X . 
\end{equation}

\begin{proposition}
  Projective measurements on $X$ in $\Rel(\Omega)$ consist of families of
  relations $\{ \Delta_{S_{i}} \}_{i \in I}$, where $\{ S_{i} \}_{i \in
    I}$ is a disjoint cover of $X$. 
\end{proposition}
\begin{proof}
  Clearly any family of relations of this form is a projective
  measurement. 
  For the converse, suppose we have a projective measurement $\{ P_{i}
  \}_{i \in I}$ on $X$. 
  The fact that $P_{i}$ is a projector in $\Rel(\Omega)$ means that
  $\Rvalue{xP_i y} = \Rvalue{y P_i x}$ and $\Rvalue{xPz} = \bigvee_y
  \Rvalue{xPy} \wedge \Rvalue{yPz}$, 
  which implies that $\Rvalue{x P_i x} \geq \Rvalue{x P_i y}$. Suppose for a
  contradiction that $\Rvalue{x P_{i} y} = \omega > \bot$ where $x \neq y$. 
  Define $R,S \colon I \to X$ by $\Rvalue{\bullet R x } = \omega = \Rvalue{\bullet S y}$, and
  $\Rvalue{\bullet R z} = \bot = \Rvalue{\bullet S z}$ for other $z$.
  Then $\Rvalue{ \bullet (P_i \circ R) x } = \bigvee_z \Rvalue{\bullet R z} \wedge
  \Rvalue{z P_i x} = \omega \wedge \Rvalue{ x P_i x} = \omega$ since
  $\Rvalue{ x P_i x} \geq \Rvalue{ y P_i x} 
  = \omega$, and also $\Rvalue{ \bullet (P_i \circ S) x } = \omega$. Similarly
  $\Rvalue{\bullet (P_i \circ R) y} = \omega = \Rvalue{\bullet (P_i \circ S) y}$, 
  and $\Rvalue{\bullet (P_i \circ R) z } = \bot =
  \Rvalue{\bullet (P_i \circ S) z}$ for other $z$. Hence $P_{i} \circ R =  P_{i} \circ S$.
  Moreover, $P_{j} \circ R = P_{j} \circ S = \bot$ for any
  $j \neq i$, by disjointness of the family, since \eg $\bot < \Rvalue{xP_j z} \leq
  \Rvalue{ xP_j x}$ implies $P_i \circ P_j \neq \bot$. Thus  $P_{k}
  \circ R = P_{k} \circ S$ for all $k \in I$, contradicting  joint
  monicity. Hence $P_{i}$ must have the form $P_{i} = \Delta_{S_{i}}$
  for some $S_{i} \subseteq X$. The fact that the family $\{ S_{i} \}_{i
    \in I}$ is a disjoint cover of $X$ now follows
  from~(\ref{dcovereq}). 
\end{proof}

Next we analyze states in $\Rel(\Omega)$. 
Firstly, we give an explicit description of the trace. If $R \colon X
\to X$ is an $\Omega$-valued relation, 
\[ \Rvalue{\bullet \Tr_X(R) \bullet } = \bigvee_x \Rvalue{ xRx }.\]
Thus the trace can be viewed as a predicate on endo-relations, which
is satisfied to the extent that the relation has a `fixpoint', \ie a reflexive
element. 

Note that $\Omega$-valued relations $R \colon I \to X$ of unit norm
correspond to $\Omega$-subsets $S$ of $X$ satisfying
\[ \bigvee_x S(x) = \top. \]
The corresponding pure state is $P_S$, defined by $\Rvalue{xP_S y} =
S(x) \wedge S(y)$.

We say that states $s$, $t$ on $X$ are \emph{equivalent} if for all
$\Omega$-subsets $S$ of $X$:
\[ \Tr_{X}(s \circ \DS) = \Tr_{X}(t \circ \DS) . \]
\begin{proposition}
\label{relpuremixedprop}
 Every state in $\Rel(\Omega)$ is equivalent to a pure state.
\end{proposition}
\begin{proof}
 If $s$ is a state on $X$, then it satisfies $\top=\bigvee_x \Rvalue{
   xsx }$, and for some relation $R$,
 \[ \Rvalue{xsy} = \bigvee_z \Rvalue{xRz} \wedge \Rvalue{ yRz} . \]
 Define an $\Omega$-subset $S=\dom(s)$ of $X$ by $x \mapsto \Rvalue{ xsx}$.
 We claim that $s$ is equivalent to $P_S$. Indeed, for any
 $\Omega$-subset $T$ of $X$,
 \begin{align*}
   \Tr_X ( s \circ \Delta_T )
   & = \bigvee_x [x (s \circ \Delta_T) x] \\
   & = \bigvee_{x,y} \Rvalue{ x \Delta_T y } \wedge \Rvalue{ y s x} \\
   & = \bigvee_x T(x) \wedge \Rvalue{ xsx} \\
   & = \bigvee_x T(x) \wedge S(x) \\
   & = \bigvee_{x,y} \Rvalue{ y P_S x } \wedge \Rvalue{ x \Delta_T y} \\
   & = \Tr_X(P_S \circ \Delta_T).
   \qedhere
 \end{align*}
\end{proof}

Finally, we consider evaluation. The scalars in $\Rel(\Omega)$ can be
identified with the locale $\Omega$. Because states
correspond to $\Omega$-subsets $S$ satisfying $\bigvee_x S(x)=\top$, and
measurements to disjoint covers, we see that  equation~(\ref{disteq}) is satisfied. 
Thus we have the following result.

\begin{proposition}
  The operational category arising from $\Rel(\Omega)$ is distributional.
\end{proposition}
\begin{proof}
  Let $\Delta_S$ be a state, and $m$ be a measurement given by a
  disjoint cover $\{S_o\}$ of $X$. Then
  \begin{align*}
    \sum_o \dev_A(\Delta_S,m)(o)
    & = \bigvee_o \Tr_X(\Delta_S \circ \Delta_{S_o}) \\
    & = \bigvee_{o,x,y} \Rvalue{ x \Delta_S y} \wedge \Rvalue{ y \Delta_{S_o} x} \\
    & = \bigvee_{o,x} S(x) \wedge S_o(x) \\
    & = \bigvee_x S(x) \wedge (\bigvee_o S_o(x)) \\
    & = \bigvee_x S(x) \\
    & = \top. \qedhere
  \end{align*}
\end{proof}

\paragraph{Discussion}
These results highlight two important differences between $\Rel(\Omega)$ and
$\Hilb$ as operational categories. In $\Hilb$, every projector can
appear as part of a projective measurement, while in $\Rel(\Omega)$ the
collective conditions of disjointness and joint monicity impose the
constraint that projectors have to be sub-identities $\Delta_S$. Moreover, in
$\Rel(\Omega)$ the distinction between \textit{superpositions} of pure states,
and \textit{convex combinations} to form mixed states, is lost, so
that every state is equivalent to a pure one. 
The relevance of this will become apparent when we discuss
non-locality in $\Rel(\Omega)$ in Section~\ref{nonlocsec}.

\section{Classical operational categories}
\label{stochmapsec}

The construction of operational representations on monoidal dagger
categories is directly inspired by quantum mechanics. However,
operational theories should also include classical physics --- or its
discrete operational residue. Our notion of operational representation
is indeed broad enough for this, as we shall now show. 

The basic classical setting we shall consider is the category
$\Stoch$. The objects are finite sets, and the morphisms $M \colon X
\to Y$ are the $X \times Y$-matrices valued in $[0, 1]$ which are
row-stochastic. Thus for each $x \in X$, we have a probability
distribution on $Y$. 

An alternative description of $\Stoch$ is as the Kleisli category for
the monad of discrete probability distributions; see~\cite{jacobs2010convexity}. 

The monoidal structure is defined as for $\Mat(S)$. Note that $\Stoch$
is not closed under matrix transposition. Indeed, we have the
following result. 

\begin{proposition}
  There is no dagger structure on $\Stoch$.
\end{proposition}
\begin{proof}
  Note that if a category $\CC$ has a dagger structure, it is in
  particular self-dual, \ie equivalent to $\CC^{\op}$. However, the
  one-element set is terminal but not initial in $\Stoch$, which is
  thus not self-dual.  
\end{proof}

It follows that we cannot directly apply the construction of
Section~\ref{cqmoprepsec}. One might consider using the formal
doubling construction on $\Stoch$ to obtain a dagger symmetric monoidal
category with a dagger trace ideal. But this would not yield the expected result;
for example, the dagger would not be given by
transpose of (bi-stochastic) matrices. However, it is easy to give a direct
definition of an operational representation, as follows. 

\begin{itemize}
  \item The sub-category $\Stoch_{\mathsf{t}}$ is defined by restricting to the
    functions (deterministic transformations), represented as matrices
    by their characteristic maps. Thus if $f \colon X \to Y$ is a
    function, for each $x \in X$ the corresponding probability
    distribution is $\delta_{f(x)}$. 

  \item A state on $X$ is a morphism $I \to X$ in $\Stoch$, or
    equivalently a probability distribution on $X$. This is the
    classical notion of mixed state. 
    The functorial action of states is described as follows. Given $f
    \colon X \to Y$, we define 
    \[ f_*(s)(y) \; = \; \sum_{f(x) = y}  s(x) . \]

  \item A measurement on $X$ is a function $m \colon X \to O$ with
    finite image $O' \subseteq O$. This is just a discrete random
    variable. The functorial action on $f \colon X \to Y$ is just 
    \[ m \mapsto m \circ f . \]

  \item The evaluation is defined by:
    \[ \dev_X(s, m)(o) \; = \; \sum_{m(x) = o} s(x) . \]
\end{itemize}

The following result is easily verified.

\begin{proposition}
  The above data specifies a probabilistic operational representation
  of $\Stoch$.
  \qed
\end{proposition}

Various generalizations of this construction are possible:
\begin{itemize}
  \item We can generalize to `distributions' over an arbitrary
    commutative semiring, as in~\cite{jacobs2010convexity}. 
    This will still yield a distributional operational representation.

  \item We can generalize to probability measures over general measure
    spaces. This amounts to using the Kleisli category of the Giry
    monad~\cite{giry1982categorical}. 
\end{itemize}

\section{Non-locality in operational categories}
\label{locsec}

Having set up a general framework for operational categories, we shall
now investigate an important foundational notion in this general
setting; namely \emph{non-locality}. 

Throughout this section, we fix a distributional operational
representation $(\CC, \Ct, \Sa, \MM, \dev)$ on a monoidal category
$\CC$. 

\subsection{Empirical models}

We shall begin by showing how probability models of the form commonly
studied in quantum information and quantum foundations can be
interpreted in the corresponding operational category. In these
models, there are $n$ agents or sites, each of which has the choice of
one of several measurement settings; and each measurement has a number
of distinct outcomes.  For each choice of a measurement setting by
each of the agents, we have a probability distribution on the joint
outcomes of the measurements. 

We shall associate objects $A_1, \ldots , A_n$ with the $n$ sites. We
define $A := A_1 \otimes \cdots \otimes A_n$. We fix a state $s \in
\Sa_A$. For each combination of measurements $(m_1, \ldots , m_n)$,
where $m_i \in \MM_{A_i}$ for $i = 1, \ldots , n$, we obtain the
measurement $m := m_1 \otimes \cdots \otimes m_n$ by inclusion of pure
tensors. Now the probability of obtaining a joint outcome $o := (o_1,
\ldots , o_n)$ for $m$ is given by  
\[ p(o|m) \; := \; \dev_A(s,m)(o) . \]

We can regard these models as observational `windows' on the
operational theory. They represent the directly  accessible
information predicted by the theory, and provide the empirical
yardstick by which it is judged. 

\subsection{Non-locality}

We now define what it means for an empirical model of the kind
described in the previous sub-section to exhibit non-locality. 
We shall follow the traditional route of using hidden variables
explicitly, although we could equivalently, and perhaps more
elegantly, formulate non-locality in terms of the (non-)existence of a
joint distribution \cite{fine1982joint,abramsky2011unified}. 

We are assuming a fixed distributional model, with a semiring of weights $\WW$.
A \emph{$\WW$-distribution} on a set $X$ is a function $d : X \rarr \WW$ of finite support, such that
\[ \sum_{x \in X} d(x) \; = \; 1 . \]
A hidden-variable model for an empirical model is defined using a set
$\Lambda$ of hidden variables, with a fixed distribution
$d$.\footnote{The assumption of a fixed distribution $d$ is
  technically the condition of `$\lambda$-independence'
  \cite{dickson1999quantum}.} 
For each $\lambda \in \Lambda$, the model specifies a 
distribution $\ql(o|m)$ on outcomes $o$ for each choice of
measurements $m$. The required  condition for the hidden variable
model to realize the empirical model $p$ is that, for all $m$ and $o$: 
\[ p(o|m) \; = \; \sum_{\lambda \in \Lambda} \ql(o|m) \cdot d(\lambda) . \]
That is, we recover the empirical probabilities by averaging over the
hidden variables. 

We say that the hidden-variable model is \emph{local} if, for all
$\lambda \in \La$, $m = (m_1, \ldots ,m_n)$, and $o := (o_1, \ldots ,
o_n)$: 
\[ \ql(o|m) \; = \; \prod_{i=1}^n \ql(o_i | m_i) . \]
Here $\ql(o_i | m_i)$ is the marginal:
\[ \ql(o_i | m_i) \; = \; \sum_{o'_i = o_i, m'_i = m_i} \ql(o' | m') . \]
We say that the empirical model $p$ is \emph{local} if it is realized by some
local hidden-variable model; and \emph{non-local} otherwise. 

Note that the definition of non-locality makes sense for any
distributional operational category. Thus we can lift these ideas to
the general level of operational categories. We say that an
operational category exhibits non-locality if it gives rise to a
non-local empirical model. Ultimately, we have a criterion for
ascribing non-locality to monoidal process categories themselves,
relative to a given distributional operational representation. 

\section{Examples of non-locality}

We shall now investigate non-locality in a number of examples.

\subsection{Hilbert spaces}

As expected, the operational category arising from $\Hilb$, which is
essentially the finite-dimensional part of standard quantum mechanics,
does exhibit non-locality. 

As a standard example --- essentially the one used by Bell in his
original proof of Bell's theorem --- consider the following table.
\begin{center}
\begin{tabular}{l|ccccc}
& $(0, 0)$ & $(1, 0)$ & $(0, 1)$ & $(1, 1)$  &  \\ \hline
$(a, b)$ & $1/2$ & $0$ & $0$ & $1/2$ & \\
$(a, b')$ & $3/8$ & $1/8$ & $1/8$ & $3/8$ & \\
$(a', b)$ & $3/8$ & $1/8$ & $1/8$ & $3/8$ &  \\
$(a', b')$ & $1/8$ & $3/8$ & $3/8$ & $1/8$ & 
\end{tabular}
\end{center}
It lists the probabilities that one of two outcomes (0 or 1) occurs
when simultaneously measured with one of two measurements at two sites
($a$ or $a'$ at the first site, and $b$ or $b'$ at the second). 
This table can be realized in quantum mechanics, \eg by a Bell state,
written in the $Z$ basis as 
\[ \frac{\mid \uparrow \uparrow  \rangle \; + \; \mid \downarrow \downarrow \rangle}{\sqrt{2}} , \]
subjected to spin measurements in the   $XY$-plane of the Bloch
sphere, at a relative angle of $\pi/3$. 

A standard argument (see \eg
\cite{bell1964einstein,abramsky2011unified}) shows that this table
cannot be realized by a local hidden-variable model. 

The same reasoning applies to C*-categories and subcategories of
$[\CC,\Hilb]$, taking the hidden variables componentwise. The constant
functor valued e.g.~at the model described above then still shows that such
operational categories are non-local.

\subsection{Relations}
\label{nonlocsec}

Suppose we are given an empirical model in the distributional operational
category obtained from $\Rel(\Omega)$. 
The types are sets $X_1, \ldots , X_n$,
there is a state $s = \Delta_S$ for a $\Omega$-subset $S$ of $X := \prod_i
X_i$ satisfying $\bigvee_x S(x)=\top$, and measurements $m_i  = \{ \Delta_{S^i_{o}} \}_{o \in O'}$,
where $\{ S^i_o \}$ is a disjoint cover of $X_i$. For each combination
of measurements $m$ and outcomes $o$, we have: 
\[ p(o|m) = \left\{ \begin{array}{ll}
\bigvee_x S(x) \wedge S_o^i(x) & \mbox{ if } o \in O', \\
0 & \mbox{ otherwise.}
\end{array}
\right.
\]

We shall now construct a local hidden-variable model which realizes
this empirical model, using the elements of $X$ as the hidden
variables. We define the distribution $\ds$ on $X$ as $x \mapsto
S(x)$. Note that we are working over $\Omega$ (the locale of scalars in
$\Rel(\Omega)$), so this is a well-defined distribution, which sums to $1$.

We define $\px(o|m) \equiv \bigwedge_i S^i_{o_i}(x_i)$, so this
hidden-variable model is local by construction.

We must verify that this model agrees with the empirical model. This
comes down to the following calculation for $o \in O'$:
\[ 
  p(o|m)
  = \bigvee_x S(x) \wedge S^i_o(x)
  = \bigvee_x S(x) \wedge \bigwedge_i S^i_{o_i}(x)
  = \bigvee_x p^x(o|m) \wedge \ds(x).
\]

We conclude from this that $\Rel(\Omega)$, despite being a `quantum-like'
monoidal dagger-category, does \textit{not} admit non-local
behaviour. 
This stands in interesting counter-point to the fact that, as shown
extensively in \cite{abramsky2010relational}, relational models can be
used to give `logical' proofs of non-locality and contextuality, in
the style of `Bell's theorem without inequalities'
\cite{greenberger1990bell}. 
The key point is that these logical proofs are based on
showing the non-existence of global sections compatible with a given
empirical model; while here we are looking at empirical models
generated by \textit{states} in $\Rel(\Omega)$, which are exactly sets of
global elements. 

The key feature of quantum mechanics, by contrast, is that quantum
states under suitable measurements \emph{are} able to realize families of
probability distributions which have no global sections. 

Bearing in mind that finite-dimensional quantum mechanics corresponds
to the operational category arising from $\Mat(\Complex)$, while
$\FRel(\Omega)$ is $\Mat(\Omega)$, this shows that \emph{idempotence of the
  scalars implies that only local behaviour can be realized}; thus
non-locality can only arise in non-idempotent situations.

\subsection{Classical stochastic maps}

We now consider the case of classical stochastic maps, as discussed in
Section~\ref{stochmapsec}. 
This is in fact quite similar to the case for $\Rel$.
Given an empirical model realized by sets $X_1, \ldots , X_n$, a state
$s$ which is a probability distribution on $X := \prod_i X_i$, and
measurements $m_i \colon X_i \to O$, we again take the hidden
variables to be the elements of $X$. We can write $s$ as a convex
combination 
\[ s = \sum_{x \in X} \mu_x \dx . \]
Note that $\mu$ is a probability distribution on $X$. 
We can define $\px(o | m) := \ddo(m(x))$. 
Clearly $\px(o|m) = \prod_i \px(o_i|m_i)$, so this hidden-variable
model is local. 

It is straightforward to verify that the probabilities $p(o|m)$ are
recovered by averaging over the deterministic hidden variables. 

Thus we conclude, as expected, that $\Stoch$ does not exhibit
non-locality. 

In fact, we can say more than this. We can calibrate the expressiveness of an operational theory in terms of which empirical models it realizes. We shall now show that $\Stoch$ realizes \emph{exactly} those models which have local hidden-variable realizations.

To see this, suppose we are given sets of measurements $M_1, \ldots , M_n$. We define $M := \bigsqcup_i M_i$, the disjoint union of these sets of measurements, and $X := O^M$.
Thus elements of $X$ simultaneously assign outcomes to all measurements.
For each $m = (m_1,\ldots,m_n) \in \prod_i M_i$, we define a map $\hat{m} \colon X \rarr O$ by
\[ \hat{m} \colon x \mapsto (x(m_1), \ldots , x(m_n)) . \]
For each $m$, we get the probability distribution on outcomes given by
\[ d_m \colon o \mapsto \sum_{\hat{m}(x) = o} s(x) . \]
This is the empirical model realized by the state $x$, viewed as a probability distribution on the hidden variables $X$; and as shown e.g.~in \cite{abramsky2011unified}, all local models are of this form.

\subsection{Signed Stochastic Maps}

We shall now consider a variant of $\Stoch$ which has much greater expressive power in terms of the empirical models it realizes.
This is the category $\SStoch$ of \emph{signed} stochastic maps; real matrices such that each row sums to 1. Thus for each input, there is a `signed probability measure' on outputs, which may include `negative probabilities' \cite{wigner1932quantum,Dirac42,moyal1949quantum,feynman1987negative}.
An operational representation can be defined for $\SStoch$ in the same fashion as for $\Stoch$; it is still distributional.

The following result can be extracted from \cite[Theorem 5.9]{abramsky2011unified}, using the same encoding of empirical models which we employed in the previous sub-section.
The reader should refer to \cite[Theorem 5.9]{abramsky2011unified} for the details, which are non-triviai.

\begin{proposition}
The class of empirical models which are realized by the operational category obtained from $\SStoch$ are exactly the no-signalling models; thus they properly contain the  quantum models.
\end{proposition}

This says that the operational category obtained from $\SStoch$ is \emph{more} expressive, in terms of the empirical models it realizes, than the canonical operational  category derived from $\Hilb$, which corresponds to quantum mechanics.

\paragraph{Example}
We consider the bipartite system with two measurements at each site, each with outcomes $\{ 0, 1 \}$. Thus the disjoint union $M$ of the two sets of measurements has four elements, and $X = \{ 0, 1 \}^M$ has 16 elements.
Now consider the following state:
\[ x \; := \; [1/2, 0, 0, 0, - 1/2, 0, 1/2, 0, - 1/2, 1/2, 0, 0, 1/2, 0, 0, 0 ] . \]
The distributions it generates for the various measurement combinations can be listed in the following table.
\begin{center}
\begin{tabular}{l|ccccc}
& $(0, 0)$ & $(1, 0)$ & $(0, 1)$ & $(1, 1)$  &  \\ \hline
$(a, b)$ & $1/2$ & $0$ & $0$ & $1/2$ & \\
$(a', b)$ & $1/2$ & $0$ & $0$ & $1/2$ & \\
$(a, b')$ & $1/2$ & $0$ & $0$ & $1/2$ & \\
$(a', b')$ & $0$ & $1/2$ & $1/2$ & $0$ & 
\end{tabular}
\end{center}
This can be recognized as the \emph{Popescu-Rohrlich box} \cite{popescu1994quantum}, which achieves super-quantum correlations.

The state $x$ can be obtained from the PR-box specification by solving a system of linear equations; see \cite{abramsky2011unified} for details.

\section{Final remarks}

This paper makes a first precise connection between monoidal
categories, and the categorical quantum mechanics framework, on the
one hand, and operational theories on the other. 
Clearly, this can be taken much further. We note a number of
directions which it would be interesting to pursue. 

\begin{itemize}
  \item We have used our framework of operational categories to study
    non-locality in a general setting. In particular, we have a clear
    definition of whether a model of categorical quantum mechanics
    exhibits non-locality or not, as explained at the end of Section~\ref{locsec}. As we saw, while Hilbert-space
    quantum mechanics does, the category of sets and relations, which
    forms a very useful `foil' model for quantum mechanics in many
    respects \cite{spekkens2007evidence,coecke2011phase}, does not. An important further direction is to apply a
    similar analysis to \textit{contextuality}, which can be seen as a
    broader phenomenon, of which non-locality is a special case. In
    \cite{abramsky2011unified}, a general setting is developed
    allowing a unified treatment of contextuality and non-locality. We
    would like to extend the present account to this setting, in which
    \textit{compatibility} of measurements is explicitly represented,
    leading to a natural sheaf-theoretic structure. 

  \item Such a development would also lead to a more satisfactory
    treatment of outcomes, in place of the somewhat clumsy device used
    in the present paper. 


  \item It would also be interesting to interpret some of the general
    results which have been proved for operational theories, relating 
    \eg to no-broadcasting \cite{barnum2007generalized}, teleportation
    \cite{barnum2008teleportation}, and information causality
    \cite{barnum2010entropy}, in our categorical framework, and
    ultimately to obtain such results for classes of monoidal process
    categories. 

  \item We would also like to examine the issue of axiomatization or
    `reconstruction' of quantum mechanics from the categorical point
    of view.

  \item There are various constructions for turning a monoidal category of `pure' states
    into one of `mixed'
    states~\cite{selinger:completelypositive,coeckeheunen:completelypositive}. 
    It would be interesting to relate these constructions to  our canonical operational categories. Similarly, there is a
    category embodying Spekkens' toy
    theory~\cite{spekkens2007evidence,coeckeedwards:spek}. It would be of interest to study the associated operational category.
\end{itemize}

Regarding related work, we note that in~\cite{barnum2010symmetry}, the
structure of the concrete category of convex operational theories is investigated. 

\paragraph{Acknowledgements}
Financial support from EPSRC 
Senior Research Fellowship EP/E052819/1 and the U.S. Office of Naval
Research Grant Number  N000141010357 is gratefully acknowledged. 
We thank Shane Mansfield for a number of useful comments, which in particular led to an improved formulation of Proposition~\ref{relpuremixedprop}.

\appendix
\newcommand{\cat}[1]{\ensuremath{\mathbf{#1}}}
\renewcommand{\id}[1][]{\ensuremath{1_{#1}}}
\newcommand{\inprod}[2]{\ensuremath{\langle #1 \mid #2 \rangle}}
\newcommand{\sxto}[1]{\xrightarrow{\smash{#1}}}
\theoremstyle{definition}
\newtheorem{example}[equation]{Example}
\newcommand{\eqcomment}[1]{\ensuremath{\tag*{\mbox{\scriptsize{#1}}}}}

\section{First notions from category theory}
 
We shall review some basic notions from category theory. For more
detailed background, see \cite{awodey2010category}. 

A \emph{category} $\CC$ has a collection of objects $A, B, C, \ldots$, and
arrows $f, g, h, \ldots$. 
Each arrow has  specified \emph{domain} and \emph{codomain} objects:
notation is $f \colon A \to B$ for an arrow $f$ with domain $A$ and
codomain $B$. The collection of all arrows with domain $A$ and
codomain $B$ is denoted as $\CC(A,B)$.
Given arrows $f \colon A \to B$ and $g \colon B \to C$, we can
form the \emph{composition} $g \circ f \colon A \to C$. Composition is
associative, and there are identity arrows $1_{A} \colon A \to A$ for
each object $A$, with $f \circ 1_{A} = f$, $1_{A} \circ g = g$, for
every $f \colon A \to B$ and $g \colon C \to A$. An arrow $f \colon A
\to B$ is called an \emph{iso(morphism)} when $f \circ f^{-1}=1_B$ and
$f^{-1} \circ f = 1_A$ for some arrow $f^{-1} \colon B \to A$.
An arrow $f \colon A \to B$ is \emph{split monic} when $g \circ f =
1_A$ for some $g \colon B \to A$, and it is \emph{split epic} when $f
\circ g = 1_A$ for some $g \colon B \to A$; by abuse of notation, we
will write $g=f^{-1}$ in both cases.

If $\CC$ is a category, we write $\CC^{\op}$ for the opposite category, with the same objects as $\CC$, and arrows $A \rarr B$ corresponding to arrows $B \rarr A$ in $\CC$.

If $\CC$ and $\DD$ are categories, a \emph{functor} $F \colon \CC \to
\DD$ assigns an object $FA$ of $\DD$ to each object $A$ of $\CC$; and
an arrow $Ff \colon FA \to FB$ of $\DD$ to every arrow $f \colon A \to B$ of
$\CC$. These assignments must preserve composition and identities:
$F(g \circ f) = F(g) \circ F(f)$, and $F(1_{A}) = 1_{FA}$. 

Given functors $F, G \colon \CC \to \DD$, a \emph{natural transformation}
$t \colon F \natarrow G$ is a family of arrows $\{ t_A \colon FA \to GA \}$
indexed by the objects of $\CC$, such that, for every $f \colon A \to B$
in $\CC$, the following naturality diagram commutes:
\[\xymatrix{
  FA \ar^-{t_A}[r] \ar_-{Ff}[d] & GA \ar^-{Gf}[d] \\
  FB \ar_-{t_B}[r] & GB
}\]
A \emph{natural isomorphism} is a natural transformation whose
components are isomorphisms. An \emph{equivalence of categories} is a
pair of functors $F \colon \CC \to \DD$ and $G \colon \DD \to \CC$
such that there are natural isomorphisms $F \circ G \cong 1_{\DD}$ and
$G \circ F \cong 1_{\CC}$.

A \emph{symmetric monoidal category} is a structure $(\CC , \otimes ,
I, \alpha, \lambda, \rho, \sigma)$ where:
\begin{itemize}
\item $\CC$ is a category;
\item $\otimes \colon \CC \times \CC \to\CC$ is a functor (\emph{tensor});
\item $I$ is a distinguished object of $\CC$ (\emph{unit});
\item $\alpha$, $\lambda$, $\rho$, $\sigma$ are natural isomorphisms (\emph{structural isos}) with components
\[ \alpha_{A, B, C} \colon A \otimes (B \otimes C) \to (A \otimes B) \otimes C \]
\[ \lambda_A \colon I \otimes A \to A \qquad \quad \rho_A \colon A \otimes I \to A \]
\[ \sigma_{A,B} \colon A \otimes B \to B \otimes A \]
such that certain coherence diagrams commute.
\end{itemize}
Products are a classical example of symmetric monoidal structure; the category
is then called \emph{Cartesian}. The symmetric monoidal structure can
also support entanglement; the category is then called \emph{compact}~\cite{abramsky2008categorical}.

Let $\CC$ and $\DD$ be symmetric monoidal categories.
A symmetric monoidal functor 
\[ (F, e, m) \colon \CC \to \DD \]
comprises
\begin{itemize}
\item a functor $F \colon \CC \to \DD$,
\item an arrow $e \colon I_{\DD}  \to FI_{\CC}$,
\item a natural transformation $m_{A, B} \colon FA \otimes FB \to F(A \otimes B)$,
\end{itemize}
subject to coherence conditions with the structural isomorphisms. The
symmetric monoidal functor is called \emph{strong} when $m$ is a
natural isomorphism.

Let $(F, e, m), (G, e', m') \colon \CC \to \DD$ be symmetric monoidal functors.
A \emph{monoidal natural transformation} between them is a natural transformation
$t \colon F \natarrow G$ such that the following diagrams commute.
\[
\xymatrix{I \ar^-{e}[r] \ar_-{e'}[dr] & FI \ar^-{t_I}[d] \\ & GI}
\qquad \qquad
\xymatrix{FA \otimes FB \ar^-{m_{A,B}}[r] \ar_-{t_A \times t_B}[d]
   & F(A \otimes B) \ar^-{t_{A \otimes B}}[d] \\
   GA \otimes GB \ar_-{m'_{A,B}}[r] & G(A \otimes B)
}
\]

\section{Trace ideals}

This appendix further studies the notion of trace ideal, introduced in
Section~\ref{subsec:additionalstructure}. It presents several
technical results that are mathematically interesting, but would break up
the flow of the main text. For example, we characterize when trace ideals
exist, and to what extent they are unique. Also, we show that we really 
need to restrict to ideals to consider traces: the category of Hilbert
spaces does not support a trace on all morphisms. As a conceptually satisfying
corollary, we derive that dual objects in the category of Hilbert
spaces are necessarily finite-dimensional. Finally, we prove
in some detail that the category of 
Hilbert spaces indeed has a
trace ideal; this was claimed in Section~\ref{subsec:additionalstructure}, 
but the details are quite subtle. 

\subsection{Existence}

The question whether a category allows a trace ideal at all can be
answered as follows.

  A subcategory $\cat{D}$ of $\cat{C}$ is called \emph{tracial} when
  endomorphisms in $\cat{C}$ factoring through $\cat{D}$ can only do so in a
  way unique up to isomorphism. More precisely: if $f_1 \colon X \to Y$, $f_2 \colon Y \to
  X$,  $f'_1 \colon X \to Y'$, $f'_2 \colon Y' \to X$ are morphisms of
  $\cat{C}$, and $Y$ and $Y'$ are objects of $\cat{D}$, and $f_2 \circ
  f_1 = f'_2 \circ f'_1$, then there is a morphism $i \colon Y \to Y'$
  in $\cat{D}$ that is either split monic or split epic, such that
  $f'_1=i \circ f_1$ and $f'_2 = f_2 \circ i^{-1}$. 
  \[\xymatrix@R-2ex{
    & Y \ar^-{f_2}[dr] \ar@{-->}@<.5ex>^-{i}[dd] \\
    X \ar^-{f_1}[ur] \ar_-{f'_1}[dr] && X \\
    & Y' \ar_-{f'_2}[ur] \ar@{-->}@<.5ex>^-{i^{-1}}[uu]
  }\]
  The category $\cat{C}$ is called \emph{traceable} when the full
  subcategory consisting of the monoidal unit $I$ is tracial.
Notice that traceability generalizes the fact, holding in any monoidal
category, that the scalars are commutative. 

\begin{proposition}\label{tracialsubcategory}
  Any dagger monoidal tracial subcategory $\cat{D}$ of $\cat{C}$ with
  a trace ideal induces a trace ideal
  \begin{align*}
    \Ideal(X) & = \{ f \in \cat{C}(X,X) \mid f=f_2 \circ
    f_1 \,\mbox{with } f_1 \colon X \to Y, f_2 \colon Y \to X
    \,\text{and } Y \,\mbox{in }\cat{D}, 
    f_1 \circ f_2 \in \Ideal_{\cat{D}} (Y) \} \\
    \Tr(f) & = \Tr_{\cat{D}}(f_1 \circ f_2)
  \end{align*}
  on $\cat{C}$.
\end{proposition}
\begin{proof}
  One directly checks that $\Ideal(X)$ is an endomorphism ideal; in
  particular $\Ideal(I)=\cat{D}(I,I)=\cat{C}(I,I)$. Because $\cat{D}$ is tracial,
  $\Tr$ is well-defined. The axioms for the trace function are also
  readily verified.
\end{proof}

\begin{theorem}\label{thm:minimaltraceideal}
  A dagger monoidal category has a unique minimal trace ideal
  \begin{align*}
    \Ideal(X) & = \{ f \colon X \to X \mid f \text{ factors through }
    I\} \\
    \Tr(f) & = b \circ a, \text{ when } f=a \circ b \text{ with } a
    \colon I \to X \text{ and } b \colon X \to I
  \end{align*}
  and hence has any trace ideal whatsoever, if and only if it is traceable.
\end{theorem}
\begin{proof}
  That the given data form a trace ideal follows from the previous
  proposition, because the full subcategory consisting of just the
  monoidal unit $I$ is certainly (totally) traced. To see that this
  trace ideal is minimal, \ie that any trace ideal must contain this
  one, follows from the first and third axioms of endomorphism ideal.
\end{proof}

As a consequence of the previous theorem, the evaluation of
measurements on pure states is completely determined by the structure of the
category, independent of the trace ideal. If $s=\psi \circ \psi^\dag$
is a pure state on $X$, and $\{P_o\}$ a measurement, then for every outcome $o$:
\[
  \Tr(s \circ P_o)
  = \Tr(\psi \circ \psi^\dag \circ P_o)
  = \Tr(\psi^\dag \circ P_o \circ \psi)
  = \psi^\dag \circ P_o \circ \psi.
\]
Therefore, the only possible freedom the choice of a trace ideal
brings comes out in behaviour on mixed states.  

\subsection{Uniqueness}

We now consider uniqueness of trace ideals. The following
proposition proves that trace ideals are a categorical invariant, in
the sense that they are preserved under equivalence. A dagger monoidal
equivalence is a pair of functors $F \colon \cat{C} 
\to \cat{D}$ and $G \colon \cat{D} \to \cat{C}$ that form an
equivalence of categories, such that $F(f^\dag)=F(f)^\dag$ and
$G(f^\dag)=G(f)^\dag$, and there are natural isomorphisms $F(I) \cong
I$, $G(I) \cong I$, $F(X \otimes Y) \cong F(X) \otimes F(Y)$ and $G(X
\otimes Y) \cong G(X) \otimes G(Y)$ that interact with the coherence
isomorphisms in the appropriate way.

\begin{proposition}
  Trace ideals are preserved under dagger monoidal equivalence:
  if $F \colon \cat{C} \to \cat{D}$ and $G \colon \cat{D} \to \cat{C}$
  are strong monoidal functors that preserve daggers and form an equivalence
  of categories, and $(\Ideal,\Tr^{\Ideal})$ is a trace ideal in 
  $\cat{C}$, then
  \begin{align*}
    \mathcal{J}(X) & = G^{-1}(\Ideal(G(X))) = \{ g \in \cat{D}(X,X) \mid
    G(g) \in \Ideal(G(X)) \}, \\        
    \Tr^{\mathcal{J}}_X(g) & = F(\Tr^{\mathcal{I}}_{G(X)}(G(g))),
  \end{align*}
  form a trace ideal in $\cat{D}$.
\end{proposition}
\begin{proof}
  First, observe that if $f \in \Ideal(X)$, and $g \colon X \to Y$ is
  an isomorphism with inverse $h$, then $\Tr(f) = \Tr(gfh)$. Then, to
  verify that $\mathcal{J}$ is an endomorphism ideal, the first
  requirement follows from functoriality of $G$; the second from the
  fact that $G$ is monoidal; and the third from fullness of $G$
  together with monoidality of $G$. It is a dagger endomorphism ideal
  because $G$ preserves daggers. Verifying that $\Tr^{\mathcal{J}}$ satisfies the
  requirements is completely analogous, except that the last 
  condition additionally uses $F(G(s)) \cong s$.
\end{proof}

However, trace ideals need not be unique. In fact, there may even be
more than one trace function making a fixed endomorphism ideal into a
trace ideal, as the following example shows.

\begin{example}\label{tracenotunique}
  A \emph{tracial state} on a C*-algebra $A$ is a linear map $\tau
  \colon A \to \Complex$ satisfying $\tau(a^*a) \geq 0$,
  $\tau(1)=1$, and $\tau(ab)=\tau(ba)$. There exists a unital
  C*-algebra $A$ with distinct tracial states $\tau \neq \tau' \colon
  A \to \Complex$~\cite{longo:crossedproduct}. 

  Make a category $\cat{C}$ as follows. Objects are natural
  numbers. There are only endomorphisms. Morphisms $0 \to 0$ are
  complex numbers; the identity is 0, and composition is
  addition. For $n\geq 1$, morphisms
  $n \to n$ are elements of the $n$-fold direct sum $A \oplus A \oplus
  \cdots \oplus A$; the identity is $(1,1,\ldots,1)$, and composition
  is pointwise multiplication. 

  We give this category a monoidal structure by letting the tensor
  product of objects $n$ and $m$ be $n+m$. If one of $n$ or $m$ is
  $0$, the action on morphisms is by scalar multiplication. For $n,m
  \geq 1$, the action on morphisms is clear. The monoidal unit is the
  object $0$. 

  Taking $\Ideal(X)$ to be all endomorphisms on $X$ certainly gives an
  endomorphism ideal. Define $\Tr_0(z)=z$, and
  $\Tr_n(a_1,\ldots,a_n)=\sum_{i=1}^n \tau(a_i)$ for $n \geq 1$. This
  satisfies all the conditions needed to make $\Ideal$ into a trace
  ideal. But the very same construction with $\tau'$ gives a different
  trace function.
\end{example}

The previous example is in stark contrast to Cartesian categories or
compact categories, where traces are unique;
see~\cite{simpsonplotkin:fixedpoint} and~\cite{hasegawa:trace},
respectively. The counterexample above is somewhat artificial, because
all morphisms are endomorphisms. It remains unclear
whether trace ideals on, for example, compact categories, are unique.

\subsection{The need for trace ideals}

We will now show that in the category $\Hilb$, there exists no trace
ideal consisting of all morphisms. More precisely, we will show that
$\Hilb$ is not an instance of the established notion of \emph{traced
monoidal category}~\cite{joyalstreetverity:traced}. This notion asks not just for
traces of all endomorphisms, but requires a `partial 
trace' of morphisms $f \colon X \otimes U \to Y \otimes U$, resulting
in a morphism $\Tr^U(f) \colon X \to
Y$. There are then several additional axioms, such as the following naturality:
\[
  \Tr^U(f) \circ g = \Tr^U(f \circ (g \otimes \id[U])) \quad\text{ for }
  f \colon X \otimes U \to Y \otimes U, g \colon X' \to X.
\]
We will now show that the monoidal category $(\Hilb,\otimes)$
cannot be traced monoidal. Subsequently, we will show that it
\emph{does} have a trace ideal. This justifies working with trace
ideals in monoidal categories instead of traced monoidal
categories. We are indebted to Peter Selinger for the following proof.

\begin{lemma}\label{lem:traceadditive}
  Suppose $(\Hilb,\otimes)$ is traced monoidal. Then $\Tr(f + g) =
  \Tr(f) + \Tr(g)$ for all endomorphisms $f,g\colon H \to H$.
\end{lemma}
\begin{proof}
  Choose an orthonormal basis $\{\ket{0},\ket{1}\}$ for
  $\Complex^2$, and write $\ket{+} = \ket{0} + \ket{1}$. Recall that
  $\Complex^2 \otimes H \cong H \oplus H$. Define $F \colon
  \Complex^2 \otimes H \to H$ via the block matrix
  $\begin{pmatrix}f&g\end{pmatrix}$. Hence $F \circ (\ket{0} \otimes
  \id[H]) = f$ and $F \circ (\ket{1} \otimes \id[H]) = g$. Now:
  \begin{align*}
    \Tr(f+g)
    & = \Tr(F \circ (\ket{+} \otimes \id[H])) \\
    & = \Tr(F) \circ \ket{+} \eqcomment{(by naturality)} \\
    & = (\Tr(F) \circ \ket{0}) + (\Tr(F) \circ \ket{1}) \\
    & = \Tr(F \circ (\ket{0} \otimes \id[H])) + \Tr(F \circ (\ket{1}
    \otimes \id[H]) \eqcomment{(by naturality)} \\
    & = \Tr(f) + \Tr(g).
  \end{align*}
  The third equality uses that composition is bilinear.
\end{proof}

\begin{theorem}\label{thm:notracehilb}
  The monoidal category $(\Hilb,\otimes)$ is not traced monoidal.
\end{theorem}
\begin{proof}
  Suppose $(\cat{Hilb},\otimes)$ was traced monoidal. Let $H$ be an
  infinite-dimensional Hilbert space. Then there exist isomorphisms $\varphi
  \colon H  \oplus \mathbb{C}\stackrel{\cong}{\to} H$ and
  $\psi\colon H\stackrel{\cong}{\to} \mathbb{C}\oplus
  H$. Write them in block matrix form as $\varphi = \begin{pmatrix}
    \varphi_1 & \varphi_2 \end{pmatrix}$ and $\psi = \begin{pmatrix}
    \psi_1 \\ \psi_2 \end{pmatrix}$. Consider the morphisms $f_1,f_2,f_3
  \colon H \oplus \mathbb{C} \oplus H \to H \oplus \mathbb{C} \oplus H$ given by the following
  block matrices.
  \[
    f_1 = \begin{pmatrix} 1 & 0 & 0 \\ 0 & 0 & 0 \\ 0 & 0 & 0 \end{pmatrix}
   \qquad
    f_2 = \begin{pmatrix} 1 & 0 & 0 \\ 0 & 1 & 0 \\ 0 & 0 & 0 \end{pmatrix}
   \qquad
    f_3 = \begin{pmatrix} 0 & 0 & 0 \\ 0 & 1 & 0 \\ 0 & 0 & 0 \end{pmatrix}
  \]
  Let $g = \varphi \oplus \psi \colon (H \oplus \mathbb{C}) \oplus H
  \to H \oplus (\mathbb{C} \oplus H)$. Then 
  \begin{align*}
    g \circ f_2 
    & = \begin{pmatrix} \varphi_1 & \varphi_2 & 0 \\ 0 & 0 & \psi_1 \\
      0 & 0 & \psi_2 \end{pmatrix} \circ \begin{pmatrix} 1 & 0 & 0 \\
      0 & 1 & 0 \\ 0 & 0 & 0 \end{pmatrix}
      = \begin{pmatrix} \varphi_1 & \varphi_2 & 0 \\ 0 & 0 & 0 \\ 0 &
        0 & 0 \end{pmatrix} \\
    & = \begin{pmatrix} 1 & 0 & 0 \\ 0 & 0 & 0 \\ 0 & 0 &
      0 \end{pmatrix} \circ \begin{pmatrix} \varphi_1 & \varphi_2 & 0 \\ 0 & 0 & \psi_1 \\
      0 & 0 & \psi_2 \end{pmatrix} = f_1 \circ g.
  \end{align*}
  Hence
  \[
    \Tr(f_1) = \Tr(f_1 \circ g \circ g^{-1}) = \Tr(g \circ f_2
    \circ g^{-1}) = \Tr(f_2 \circ g^{-1} \circ g) = \Tr(f_2).
  \]
  But $\Tr(f_2)=\Tr(f_1+f_3)=\Tr(f_1)+\Tr(f_3)$ by
  Lemma~\ref{lem:traceadditive}. 
  And because $f_3$ has finite rank, we know that
  $\Tr(f_3)=\Tr(\id[\Complex])=1$.
  Thus $\Tr(f_2) = \Tr(f_2) + 1$, which is a contradiction.
\end{proof}

\subsection{Dual objects in $\Hilb$ are finite-dimensional}

The previous theorem allows an interesting corollary. Recall that the
main characteristic of compact categories is that objects have duals:
objects $L,R$ in a monoidal category are called \emph{dual} when 
there are maps $\eta \colon I \to R \otimes L$ and $\varepsilon \colon
L \otimes R \to I$ making the following two composites identities.
\begin{align*}
  L \cong L \otimes I \sxto{\id \otimes \eta} L \otimes (R 
  \otimes L) \cong (L \otimes R) \otimes L \sxto{\varepsilon \otimes
    \id} I \otimes L \cong L \\
  R \cong I \otimes R \sxto{\eta \otimes \id} (R \otimes L) \otimes R
  \cong R \otimes (L \otimes R) \sxto{\id \otimes \varepsilon} R
  \otimes I \cong R
\end{align*}
It is well-known that if $H \in \Hilb$ is finite-dimensional,
then $H$ and $H^*$ are dual objects by $\eta(1) = \sum_{i=1}^n \ket{i}
\otimes \bra{i}$ and $\varepsilon(\ket{i}) =  1$, for any choice of
orthonormal basis $\{\ket{i}\}_{i=1,\ldots,n}$ for $H$;
see~\cite{kellylaplaza:compactcategories,joyalstreetverity:traced,abramsky2008categorical}. This recipe does not work when $H$ is infinite-dimensional, because $\sum_i
\ket{i}$ does not converge in that case. However, this does not
exclude the possibility that there might be other
$H^*,\eta,\varepsilon$ making $H$ into a dual object. No rigorous
proof that infinite-dimensional Hilbert spaces cannot have duals has
been published, as far as we know. 

\begin{corollary}
  Objects in $(\Hilb,\otimes)$ with duals are precisely
  finite-dimensional Hilbert spaces.
\end{corollary}
\begin{proof}
  Let $H$ be an infinite-dimensional Hilbert space.
  Suppose $H$ has a dual object $H^*$. For $f \colon H \to H$,
  define $\Tr^H(f)$ as the following composite. 
  \[
  I \sxto{\eta} H^* \otimes H \cong H \otimes H^* \sxto{f \otimes
    \id[H^*]} H \otimes H^* \sxto{\varepsilon} I
  \]
  This satisfies all equations for a trace function, as far as these
  make sense `locally', for just one object $H$. In $\cat{Hilb}$, the
  object $\Complex$ always has a dual, and if $H$ and $K$ have duals,
  then so does $H \oplus K$. 
  Now, notice that the proof of Theorem~\ref{thm:notracehilb} only
  uses the trace properties `locally', \ie\ for the objects $\Complex,
  H, \Complex^2 \otimes H \cong H \oplus H, H \oplus \Complex, H
  \oplus \Complex \oplus H$. Hence the contradiction it results in
  holds here, too.
\end{proof}

In fact, in any monoidal category with biproducts, one can show that
if $A \cong A \oplus I$, then $\Tr^A(\id[A]) = \Tr^A(\id[A]) + 1$.
We thank Jamie Vicary for this observation.

\subsection{Trace class maps form a trace ideal in $\Hilb$}

To show that the usual trace of continuous linear maps between
Hilbert
spaces does in fact give a trace ideal requires some work, as virtually all 
textbooks only consider endomorphisms, whereas
the defining conditions of trace ideals also involve morphisms between
different objects.

We need to recall some terminology; for any unexplained terms, we
refer to~\cite{blankexnerhavlicek:hilbert}. Other good references
are~\cite{garling:inequalities, simon:traceideals}.
A linear map $f \colon H \to K$ between Hilbert spaces is
\emph{Hilbert-Schmidt} when $\sum_n \|f(e_n)\|_K^2 < \infty$ for an
orthonormal basis $(e_n)$ of $H$. A positive continuous linear map $f \colon H \to
H$ is \emph{trace class} when $\sum_n
\big|\inprod{e_n}{f(e_n)} \big|< \infty$ for an orthonormal basis
$(e_n)$ of $H$. An arbitrary continuous linear map $f\colon H \to H$
is trace class when its absolute value $|f| \colon H \to H$ is trace class.
Both definitions are independent of the choice of basis
$(e_n)$. If $f$ is trace class, then $\inprod{e_n}{f(e_n)}$ is
absolutely summable, and hence the following \emph{trace property} holds:
\[
  \Tr(f)=\sum_n \inprod{e_n}{f(e_n)}
\]
is a well-defined complex number. The \emph{Cauchy-Schwarz inequality}
states that 
\[
  \big| \inprod{x}{y} \big| \leq \big( \|x\|^2 \cdot \|y\|^2 \big)^{1/2}
\]
for any two elements $x,y$ of a Hilbert space. The \emph{H{\"o}lder
inequality} states that 
\[
  \sum_n |x_n \cdot y_n| \leq \big(\sum_n |x_n|^2\big)^{1/2} \cdot
  \big(\sum_n |y_n|^2\big)^{1/2} 
\]
for any two sequences $(x_n)$ and $(y_n)$ of complex numbers
with $\sum_n |x_n|^2<\infty$ and $\sum_n |y_n|^2<\infty$. 

\begin{lemma}\label{lem:hilbertschmidt}
  Let ${\xymatrix@1{H \ar@<.75ex>^-{f}[r] & K \ar^-{g}[l]}}$
  be morphisms in $\Hilb$. Then $g \circ f$ is trace class if and only
  if $f$ and $g$ are Hilbert-Schmidt.
\end{lemma}
\begin{proof}
  By polar decomposition, there is a unique partial isometry $w \colon
  H \to K$ satisfying $g \circ f = w \circ |g \circ f|$ and
  $\ker(w)=\ker(g\circ f)$.  It follows that $|g\circ f|= w^\dag \circ
  g \circ f$. Hence, for an orthonormal basis $(e_n)$ of $H$,   
  \begin{align*}
    \sum_n \big| \inprod{e_n}{|g \circ f|(e_n)} \big|
    & = \sum_n \big| \inprod{e_n}{w^\dag \circ g \circ f(e_n)} \big| \\
    & = \sum_n \big| \inprod{g^\dag \circ w(e_n)}{f(e_n)} \big| \\
    & \leq \sum_n \big( \|g^\dag
    \circ w(e_n) \|^2 \cdot \|f(e_n)\|^2\big)^{1/2}  
    \eqcomment{(by Cauchy-Schwarz)}\\
    & = \sum_n \big| \|g^\dag \circ w(e_n) \| \cdot \|f(e_n)\| \big| \\
    & \leq \big( \sum_n \|g^\dag \circ w(e_n) \|^2 \big)^{1/2} \cdot
               \big( \sum_n \|f(e_n) \|^2 \big)^{1/2}. \eqcomment{(by
                 H{\"o}lder)}
  \end{align*}
  Therefore $gf$ is trace class if and only if $\sum_n
  \|f(e_n)\|^2<\infty$ and $\sum_n \|g^\dag\circ
  w(e_n)\|^2<\infty$. Because $w$ is a partial isometry, the 
  latter inequality holds if and only if $\sum_n
  \|g(e_n)\|^2<\infty$. That is, $g\circ f$ is trace class if and only if
  $f$ and $g$ are Hilbert-Schmidt.
\end{proof}

\begin{proposition}
  The category $\Hilb$ has a dagger trace ideal consisting of the
  usual trace class maps and the usual trace function. 
\end{proposition}
\begin{proof}
  That the trace class maps on a Hilbert space $H$ are closed under
  adjoint and tensor products is easily seen. Also, any morphism
  $\Complex \to \Complex$ is trivially trace class. Now suppose that
  $f \colon H \to H$ is trace class. By the previous lemma, we
  can write $f=f_2 \circ f_1$ for Hilbert-Schmidt maps $f_i$. If $g \colon H
  \to K$ and $h \colon K \to H$ are arbitrary morphisms, then $g\circ f_2$
  and $f_1\circ h$ are again Hilbert-Schmidt. Therefore, by the previous
  lemma again, $g\circ f\circ h=(g\circ f_2)\circ (f_1\circ h)$ is
  trace class. Thus trace class maps indeed form an endomorphism ideal.

  One easily sees from the trace property that trace is the identity on scalars, is
  multiplicative on tensor products, and preserves daggers. To prove
  that $\Tr(g\circ f)=\Tr(f\circ g)$ for $f \colon H \to K$ and $g \colon K \to H$
  with both $f\circ g$ and $g\circ f$ trace class, we rely on \emph{Lidskii's
  trace formula} for separable $H$: if $h$ is trace class, then ($\sum_n
  \lambda_n(h)$ is absolutely convergent and) $\Tr(h)=\sum_n
  \lambda_n(h)$, where $\lambda_n(h)$ are the eigenvalues counted up
  to algebraic multiplicity~\cite[Theorem~3.7]{simon:traceideals}. But $g\circ f$ and $f\circ g$ have precisely the same 
  spectrum, so that $\Tr(g\circ f)=\sum_n \lambda_n(g\circ f) = \sum_n
  \lambda_n(f\circ g)=\Tr(f\circ g)$. 

  Finally, we claim that for positive, trace class functions $h \colon
  H \to H$ on any (possibly nonseparable) Hilbert space $H$, Lidskii's
  formula still holds, which finishes the proof that trace class
  operators form a trace ideal, because we may then replace $g \circ
  f$ and $f \circ g$ above by their absolute value.
  Pick an orthonormal basis $\{e_i\}$ for $H$. Since $h$ is trace
  class, $\sum_i \inprod{e_i}{h(e_i)} = \sum_i \| \sqrt{h}(e_i) \|$ is
  summable. Hence $\ker(h)^\perp=\ker(\sqrt{h})^\perp$ can only
  contain countably many 
  $e_i$. Because $h$ is positive, its range is $\ker(h^\dag)^\perp =
  \ker(h)^\perp$. Thus $h \colon H \to H$ restricts to a function $h
  \colon \ker(h)^\perp \to \ker(h)^\perp$ on a separable space.
\end{proof}

We have written the above example out in more detail than the reader
might have thought necessary, because it is easy to overlook
subtleties. For example, it is not true that if $f \colon H \to K$ and
$g \colon K \to H$ are morphisms such that $g\circ f$ is trace class, then
$f\circ g$ is trace class, too. For a counterexample, let
$H=K=\ell^2(\Nat)$, and define $f(x,y)=(0,x)$ and
$g(x,y)=(x,0)$. Then certainly $g\circ f=0$ is trace class. But it is easy
to see that $f^\dag(x,y) = (y,0)$, that $g=g^\dag=g^\dag \circ g$, and hence
that $g=g^\dag \circ g=(f\circ g)^\dag \circ (f\circ g) \geq
0$. Therefore $|f\circ g|=g$, and
\[
     \Tr(f\circ g) 
  = \sum_{m,n} \langle |f\circ g|(e_m,e_n) \mid (e_m,e_n) \rangle
  = \sum_{m,n} \langle e_m \mid e_m \rangle + \langle 0 \mid e_n\rangle 
  = \dim(H) = \infty,
\]
so that $f\circ g$ is not trace class.

\bibliographystyle{plain}
\bibliography{bdbib}

\end{document}